\documentclass[11pt,twoside]{article}
\usepackage{fancyhdr}
\usepackage[colorlinks,citecolor=blue,urlcolor=blue,linkcolor=blue,bookmarks=false]{hyperref}
\usepackage{amsfonts,epsfig,graphicx}
\usepackage{afterpage}
\usepackage{comment}
\usepackage{amsmath,amssymb,amsthm} 
\usepackage{mathtools}
\usepackage{enumerate}
\usepackage{fullpage}
\usepackage[T1]{fontenc} 
\usepackage{epsf} 
\usepackage{graphics} 
\usepackage{amsfonts,amsmath}
\usepackage{natbib} 
\usepackage{psfrag,xspace}
\usepackage{color,etoolbox, xcolor}
\usepackage{subcaption} 
\usepackage{dsfont}
\usepackage[page]{appendix}
\usepackage{soul} 

\def\ind{\perp\!\!\!\perp}

\DeclareMathOperator*{\argmax}{arg\,max}

\DeclareSymbolFont{bbold}{U}{bbold}{m}{n}
\DeclareSymbolFontAlphabet{\mathbbold}{bbold}

\newtheorem{theorem}{Theorem}
\newtheorem{lemma}{Lemma}
\newtheorem{corollary}{Corollary}
\newtheorem{proposition}{Proposition}
\newtheorem{algorithm}{Algorithm}
\newtheorem{assumption}{Assumption}
\newtheorem{remark}{Remark}
\theoremstyle{definition}
\newtheorem{definition}{Definition}

\theoremstyle{remark}

\newcommand\independent{\protect\mathpalette{\protect\independenT}{\perp}}
\def\independenT#1#2{\mathrel{\rlap{$#1#2$}\mkern2mu{#1#2}}}

\begin{document}

\def\spacingset#1{\renewcommand{\baselinestretch}%
{#1}\small\normalsize} \spacingset{1}

\raggedbottom
\allowdisplaybreaks[1]


  
  \title{\vspace*{-.4in} {Intervention effects based on potential benefit}}
  \author{Alexander W. Levis$^1$, Eli Ben-Michael$^{1,2}$, Edward H. Kennedy$^1$ \\  \\ \\
    $^1$Department of Statistics \& Data Science, \\
    Carnegie Mellon University \\
    $^2$Heinz College of Information Systems and Public Policy, \\
    Carnegie Mellon University \\ \\ 
    \texttt{\{alevis, ebenmichael\}  @ cmu.edu}; \\ \texttt{edward@stat.cmu.edu} \\
\date{}
    }
    
  \maketitle
  \thispagestyle{empty}

\begin{abstract}
Optimal treatment rules are mappings from individual patient characteristics to tailored treatment assignments that maximize mean outcomes. In this work, we introduce a conditional \textit{potential benefit} (CPB) metric that measures the expected improvement under an optimally chosen treatment compared to the status quo, within covariate strata. The potential benefit combines (i) the magnitude of the treatment effect, and (ii) the propensity for subjects to naturally select a suboptimal treatment. As a consequence, heterogeneity in the CPB can provide key insights into the mechanism by which a treatment acts and/or highlight potential barriers to treatment access or adverse effects. Moreover, we demonstrate that CPB is the natural prioritization score for individualized treatment policies when intervention capacity is constrained. That is, in the resource-limited setting where treatment options are freely accessible, but the ability to intervene on a portion of the target population is constrained (e.g., if the population is large, and follow-up and encouragement of treatment uptake is labor-intensive), targeting subjects with highest CPB maximizes the mean outcome. Focusing on this resource-limited setting, we derive formulas for optimal constrained treatment rules, and for any given budget, quantify the loss compared to the optimal unconstrained rule. We describe sufficient identification assumptions, and propose nonparametric, robust, and efficient estimators of the proposed quantities emerging from our framework. Finally, we illustrate our methodology using data from a prospective cohort study in which we assess the impact of intensive care unit transfer on mortality.
\end{abstract}

\bigskip


\pagebreak

\section{Introduction}
The effects of exposures or treatments can vary substantially across members of a population. Consequently, it is important to consider treatment strategies that take into account individual-level characteristics in order to optimize outcomes. These strategies are referred to as \textit{treatment rules}, \textit{policies}, or \textit{dynamic treatment regimes}, and there has been substantial work on characterizing, identifying, and estimating optimal treatment rules from randomized and observational data, in both time-fixed and longitudinal settings \citep{murphy2003, robins2004, zhao2012}. See \citet{chakraborty2013} and \citet{tsiatis2019} for textbook coverage of this topic.

In the absence of any constraints, optimal treatment rules are characterized entirely by conditional (on measured covariates) average potential outcomes under the possible treatment values \citep{murphy2003}, e.g., for a point binary treatment $A \in \{0,1\}$, the optimal rule assigns active treatment ($A = 1$) when the conditional average treatment effect is positive, and otherwise assigns the control condition ($A = 0$). In practice, though, there is substantial interest in the interplay between putative treatment effects and selection into treatment level. For instance, in the extensive literature on the economic return of college education, there has been debate as to whether there is \textit{positive selection} or \textit{negative selection} on returns from education,
i.e., whether the individuals most likely or least likely, respectively, to obtain a college education are those that benefit most from attending college on average \cite[see, e.g.][]{willis1979, carneiro2001, carneiro2003, carneiro2011, brand2010, zhou2020}. Distinguishing between these two hypotheses has implications for policy: should educational opportunity expansion efforts focus on all individuals equally, on those more likely to attend college, or on those less likely to attend college \citep{brand2010}?

A single metric encompassing the net potential benefit of an intervention relative to the status quo would provide the means to justify such policy decisions, and is a key contribution of this work; to our knowledge, such a potential benefit metric has not yet been described in the literature on health policy or dynamic treatment regimes in causal inference. Beyond policy considerations, a potential benefit metric is a valuable descriptive tool for quantifying how far the status quo is from the optimal treatment assignment in terms of average outcomes, marginally or in key subgroups.

In this work, we consider an important, unstudied policy setting in which treatments are widely accessible to individuals in a population of interest, but there is limited capacity for the intervening party to reach out and induce uptake of a given treatment option, i.e., the principal cost arises from intervening, not from acquiring treatment itself. 
In general, it may be that both levels of a binary treatment, say, may be beneficial in disparate subgroups, and/or that individuals in some subgroups are more or less likely to obtain their optimal level of treatment. An outreach intervention designed to induce optimal treatments in these subpopulations, prioritizing groups with higher treatment effects and lower propensities for seeking out optimal treamtent, would be particularly appealing. It is therefore of interest to characterize when and how experimental or real-world data may be used to best design such outreach interventions, especially when the population of interest is very large and only a subset can reasonably be targeted.

\subsection{Related work and contributions}
For the most part, methodologists studying optimal treatment rules have not incorporated constraints---arising for example from budget or resource limitations---that arise in the real world and may well affect the implementation of any given policy. When one is unable to freely implement an idealized policy on an entire population (or subsample thereof), it becomes advisable to target high-priority individuals so as to maximize expected outcomes under the limitations at hand. Responding to this gap, there has been a recent push to formally incorporate different kinds of real-world constraints. \citet{luedtke2016b} described a setting in which there is limited supply in one of two treatment possibilities, so that it could not be given to all individuals even if it were beneficial to everyone. These authors characterized the optimal treatment rule given that a fixed maximal proportion of the population can be given the limited treatment, and developed efficient estimators of the expected counterfactual outcome under this rule. Under the same constraint, \citet{qiu2021} extended these ideas to the case where one intervenes on \textit{encouragement} to treatment, i.e., the intervention is described with respect to an instrumental variable which has a causal influence on the treatment actually taken. \citet{sun2021} and \citet{qiu2022} instead consider a setting where there is a random cost, dependent on both treatment choice and baseline covariates, associated with undergoing treatment, and developed optimal treatment rules that respect a given budget for the expected cost.

Responding to gaps outlined above, we make several contributions. We begin by providing a general definition of the potential benefit of a targeted intervention, based on an optimal unconstrained treatment rule. We show that the potential benefit incorporates both the magnitude of the conditional average treatment effect, as well as the real-world propensities of subjects in the population of interest to seek out their optimal treatment without intervention. We then use this potential benefit metric to formally study policies that jointly select a subset of the population on which to intervene---constraining this subset to be of a fixed maximal size---and tailor treatment selection on this subpopulation.
In addition to characterizing optimal treatment rules in this setting, we quantify the gap in the mean counterfactual outcome between the overall optimal constrained and unconstrained policies, and develop nonparametric efficient and robust estimators of the mean counterfactual under these rules, as well as other related quantities that emerge from our framework. Together, the proposed methodology can be applied to data from any (existing or new) observational study, yielding descriptive insight into the interplay between treatment effect heterogeneity and treatment allocation, and providing a basis for intelligent design of optimal policies or interventions that respect budget constraints.

The remainder of this article is organized as follows. In Section~\ref{sec:pot-benefit}, we propose a general definition for conditional potential benefit (CPB), and discuss its interpretation and utility. In Section~\ref{sec:policy}, we (a) define the class of policies motivated by the CPB, and formalize a novel resource constraint on the intervention capacity; (b) introduce quantities related to mean counterfactuals under the policies of interest, characterize optimal constrained regimes, and present relevant identifying assumptions; and (c) illustrate the framework in two synthetic examples. In Section~\ref{sec:estimation}, we develop efficient and robust estimators for the statistical functionals representing the causal quantitities of interest. In Section~\ref{sec:applications}, we demonstrate the proposed framework in a cohort study assessing the effect of intensive care unit transfer on mortality. Finally, in Section~\ref{sec:discussion}, we discuss possible extensions and variants, and provide concluding remarks.

\section{Optimal Unconstrained Policies and Potential Benefit} \label{sec:pot-benefit}

Suppose we wish to make treatment decisions on the basis of a vector $X \in \mathcal{X} \subseteq \mathbb{R}^d$ of baseline covariates. We consider for now an arbitrary treatment $A \in \mathcal{T}$---though we focus later on the case $\mathcal{T} = \{0,1\}$---and a real-valued outcome of interest $Y \in \mathcal{Y} \subseteq \mathbb{R}$, where larger values are favorable. A \textit{treatment rule} or \textit{policy} is a function $d: \mathcal{X} \to \mathcal{T}$, potentially with additional stochastic input, that assigns treatment for a subject based on their covariates $X$. For a given rule $d$, we denote $Y(d)$ as the \textit{potential outcome} that would be observed had treatment been set according to the policy $d$; for simplicity we write $Y(a)$ for the static rule setting $A = a$ deterministically, for $a \in \mathcal{T}$. In general, we will evaluate a policy $d$ based on its \textit{value}, i.e., the mean of the counterfactual under that rule, $\mathbb{E}(Y(d))$. The \textit{natural} values of treatment and outcome, denoted $A$ and $Y$, respectively, are those that would be observed in the absence of any intervention.

An \textit{optimal} treatment rule is one that attains the maximal possible value, and one can find such a rule by maximizing the conditional mean counterfactual outcome \citep{murphy2003, robins2004}: formally, $h^* : \mathcal{X} \to \mathcal{T}$ is an optimal treatment rule for maximizing the marginal mean outcome if and only if
\[h^*(X) \in \argmax_{a \in \mathcal{T}} \, \mu_a^{\dagger}(X), \text{ with probability } 1,\]
where $\mu_a^{\dagger}(X) = \mathbb{E}(Y(a)\mid X)$, for $a \in \mathcal{T}$. Based on an optimal treatment rule, we define the potential benefit of treatment in terms of the outcome under an optimal treatment and the observed, natural value of the outcome.

\begin{definition}\label{def:pot-benefit}
    The conditional potential benefit (CPB) is defined as 
    \[\beta^{\dagger}(X) \coloneqq \mathbb{E}(Y(h^*) - Y \mid X),\]
    representing the gap (in terms of mean outcomes) between an optimal treatment rule and the natural treatment selection process, i.e., that which gives rise to the natural outcome value $Y$.
\end{definition}

As an example, for binary treatment $A \in \{0,1\}$, an optimal treatment rule is given by $h^*(X) = \mathds{1}(\tau^{\dagger}(X) > 0)$, where $\tau^{\dagger}(X) = \mu_1^{\dagger}(X) - \mu_0^{\dagger}(X)$, and the CPB is given by
\[\beta^{\dagger}(X) = \mathbb{E}(\{h^* - A\}\{Y(1) - Y(0)\} \mid X) = h^*(X)\tau^{\dagger}(X) + \mu_0^{\dagger}(X) - \mathbb{E}(Y \mid X).\]
From this last equation, we see that potential benefit depends both on the treatment effect, through $\{Y(1) - Y(0)\}$, and through the natural propensity for suboptimal treatment selection, through $\{h^* - A\}$. For a general discrete treatment $A \in \mathcal{T}$, the CPB is given by $\beta^{\dagger}(X) = \sum_{a \in \mathcal{T}} \mathds{1}(h^*(X) = a) \sum_{a' \neq a}\mathbb{E}(Y(a) - Y(a') \mid X, A = a')\mathbb{P}[A = a' \mid X]$.
For simplicity, we will focus on binary treatments, but we
emphasize that Definition~\ref{def:pot-benefit} is completely general and applies equally to continuous treatments.

Identification and estimation of $\beta^{\dagger}$ are discussed in Sections~\ref{sec:policy} and \ref{sec:estimation}, respectively. First, though, some general remarks regarding the interpretation and utility of the CPB measure are warranted. From Definition~\ref{def:pot-benefit}, we see that the CPB quantifies suboptimality of the status quo regime (i.e., yielding the natural value $A$), which can help flag key subgroups. A policy-maker with access to the function $\beta^{\dagger}$ would be able to identify individuals with the most expected to gain from tailored treatment. Even beyond policy implications, the CPB provides a useful descriptive tool: by identifying subgroups for which there is a large treatment effect and/or the observed treatment selection is often suboptimal, practitioners can gain insights into their specific scientific problem. Namely, the CPB can motivate mechanistic hypotheses with regards to the action of the treatment itself (e.g., if there are higher/lower treatment effects in certain subgroups), or suggest that there are possible adverse effects or barriers to treatment access (e.g., explaining suboptimal treatment selection).

Another natural consequence of the definition of the CPB is that it can be useful for designing policies that respect real-world constraints. In Section~\ref{sec:policy}, we describe a novel class of treatment rules that respect constraints on the proportion of the population that can be targeted for an intervention. Moreover, we show that optimal constrained policies are those that assign $h^*$ to a subpopulation with the highest CPB values. In other words, $\beta^{\dagger}$ represents the priority score that should be used when there are constraints on intervention capacity.

\section{Optimal Policies under Intervention Constraints} \label{sec:policy}

\subsection{Contact rule framework} 
\label{sec:framework}
For the remainder of the paper, let $A \in \{0,1\}$ be a binary-valued treatment, and let $\pi(X) = \mathbb{P}[A = 1 \mid X]$ denote the propensity score, i.e., the natural treatment mechanism. To capture policies which aim to intervene only on a subset of the population, we introduce the notion of a \textit{contact rule}, $\Delta: \mathcal{X} \to [0,1]$, which assigns a probability of contacting and intervening on a subject based on their covariates $X$. In practice, we are often constrained to consider contact rules $\Delta$ that contact at most $\delta$\% of the population (i.e., $\mathbb{E}(\Delta(X)) \leq \delta$), for some $\delta \in [0,1]$. Given a contact rule $\Delta$, together with a putative ``good'' policy $h: \mathcal{X} \to \{0,1\}$, we will study treatment rules of the form
\begin{equation}\label{eq:contact-policy}
    d(\Delta, h) = H_{\Delta}h + (1 - H_{\Delta})A,
\end{equation}
where $H_{\Delta} \overset{d}{=} \mathrm{Bernoulli}(\Delta(X))$ and $H_{\Delta}  \ind (A, Y(0), Y(1)) \mid X$. In words, $d(\Delta, h)$ assigns $h$ to those contacted (i.e., where $H_{\Delta} = 1)$, and otherwise the natural value of treatment occurs.

In practice, it may be that decision makers have access only to some subset of covariates, $W \subseteq X$, e.g., if $X$ contain richer information collected in a study context. In Appendix~\ref{app:subset}, we provide an in-depth study of treatment rules that are constrained to only depend on $W$.


\subsection{Estimands and identification}\label{sec:estimands}
We suppose that we observe a random sample of $n$ copies of $O = (X, A, Y) \sim \mathbb{P}$, where $X \in \mathcal{X}$ are baseline covariates as above, $A \in \{0,1\}$ is the observed treatment level and $Y \in \mathcal{Y}$ is the observed outcome variable. These may arise as experimental or observational data, where in the latter case $X$ would include relevant confounders. The sample $O_1, \ldots, O_n$ represents ``batch'' data on which treatment rules will be learned; we envision implementing these rules in the same population $\mathbb{P}$ (or some large subsample theoreof).

Our first goal is to characterize optimal treatment rules of the form~\eqref{eq:contact-policy}. That is, given some constraint on the proportion of the population on which we can feasibly intervene, we wish to determine the optimal subpopulation (through $\Delta$) on which to intervene, and the optimal treatment rule (through $h$) to apply on this subpopulation, so as to maximize average potential outcomes. In order to link the natural value of the outcome to the potential outcomes under various policies, we will make the following assumption:

\begin{assumption}[Consistency] \label{ass:consistency}
    $Y(A) = Y$, and $Y(d) = dY(1) + (1 - d)Y(0)$, for all treatment rules $d: \mathcal{X} \to \{0,1\}$.
\end{assumption}
Assumption~\ref{ass:consistency} requires that intervention on $A$ is well-defined (i.e., there is only one version of treatment) and that outcomes are not affected by treatments of other subjects. For point identification of effects, we will further rely on the following standard assumptions:
\begin{assumption}[Positivity] \label{ass:positivity}
    $\mathbb{P}[A = 1 \mid X] \in [\epsilon, 1 - \epsilon]$ with probability 1, for some $\epsilon > 0$.
\end{assumption}
\begin{assumption}[No unmeasured confounding] \label{ass:NUC}
    $A \independent Y(a) \mid X$.
\end{assumption}
Note that Assumption~\ref{ass:positivity} requires that no subjects are deterministically receiving one particular treatment level, while Assumption~\ref{ass:NUC} requires that treatment is as good as randomized within levels of $X$. Define, when Assumption~\ref{ass:positivity} holds, the observational mean outcomes $\mu_a(X) = \mathbb{E}(Y \mid X,  A = a)$, for $a \in \{0,1\}$, and the contrast $\tau(X) = \mu_1(X) - \mu_0(X)$. Henceforth, to reduce clutter, we will often omit inputs to functions when the context makes the input clear, e.g., $\tau \equiv \tau(X)$.


By standard arguments, under Assumptions~\ref{ass:consistency}--\ref{ass:NUC}, $\mu_a^{\dagger} \equiv \mu_a$, so that $\tau^{\dagger} \equiv \tau$. Consequently, these assumptions are sufficient for identification of the quantities introduced in Section~\ref{sec:pot-benefit}: the optimal unconstrained rule, $h^* \equiv \mathds{1}(\tau > 0)$, its conditional value, $\mathbb{E}(Y(h^*)\mid X) = h^* \tau + \mu_0$, and the CPB, $\beta^{\dagger} \equiv \beta \coloneqq \tau(h^* - \pi)$. Further, define $c(X) \coloneqq h^*(X)(1 - \pi(X)) + (1 - h^*(X))\pi(X) =  \mathbb{P}[A \neq h^*(X) \mid X]$ as the conditional probability that a subject does not naturally receive their optimal treatment value $h^*(X)$.
Then we can write  the CPB as $\beta \equiv c |\tau|$, i.e., $\beta(X)$ is the average cost of making the wrong treatment decision, $|\tau(X)|$, weighted by the probability of that wrong decision occurring, $c(X)$. 


The following is a preliminary result for identification of the value of a given treatment rule $d(\Delta, h)$. The proofs for this result and all others are given in the Appendices.

\begin{proposition}\label{prop:ident}
    Under Assumptions~\ref{ass:consistency}--\ref{ass:NUC}, and given a contact rule $\Delta: \mathcal{X} \to [0,1]$ and a policy $h:\mathcal{X}\to\{0,1\}$, the conditional value of the policy $d(\Delta, h)$, defined by~\eqref{eq:contact-policy}, is given by
    \[\mathbb{E}[Y(d(\Delta, h)) \mid X]= \Delta\tau (h - \pi) + m,\]
    where $m(X) = \mathbb{E}(Y \mid X)$. Thus, the marginal value of the policy $d(\Delta, h)$ is 
    $\mathbb{E}[Y(d(\Delta, h))] = \mathbb{E}\left(\Delta\tau (h - \pi) + Y\right)$.
\end{proposition}

While Proposition~\ref{prop:ident} gives point identification, in Appendix~\ref{app:sens}, we provide a more general result relying only on Assumption~\ref{ass:consistency}, expressing the policy value as a function of $\mu_0^{\dagger}, \mu_1^{\dagger}, m$, and the distribution of $X$. We then use this result to develop sharp bounds for these policy values under sensitivity models quantifying violations of Assumption~\ref{ass:NUC}.

Given the result of Proposition~\ref{prop:ident}, under Assumptions~\ref{ass:consistency}--\ref{ass:NUC}, we can characterize an optimal pair $(\Delta, h)$ that maximizes the value $\mathbb{E}[Y(d(\Delta, h))]$, subject to the constraint $\mathbb{E}(\Delta(X)) \leq \delta$, for a fixed $\delta \in [0,1]$.

\begin{proposition}\label{prop:optimal}
Consider the optimization problem
    \begin{align}
    \begin{split}\label{eq:opt-policy}
        \max_{h, \Delta} \; & \mathbb{E}[Y(d(\Delta, h))] \\
        \text{subject to } \; & \mathbb{E}(\Delta(X)) \leq \delta 
    \end{split}
    \end{align}
    Under Assumptions~\ref{ass:consistency}--\ref{ass:NUC}, and assuming the $(1 - \delta)$-quantile of $\beta$ is unique, denoted as $q_{1 - \delta}$, the contact rule $\Delta_{\delta}^*(X) = \mathds{1}(\beta(X) > q_{1 - \delta})$ and the optimal unconstrained policy $h^*$ jointly solve~\eqref{eq:opt-policy}, achieving the maximal value $\mathbb{E}[Y(d(\Delta_{\delta}^*, h^*))] = \mathbb{E}(\Delta_{\delta}^* \beta + Y) \eqqcolon V_{\delta}$.
\end{proposition}

According to the preceding result, if one can only intervene on $\delta$ proportion of the population, the CPB (i.e., $\beta$) represents a priority score on which the subjects above the $(1 - \delta)$-quantile should be targeted. Moreover, the optimal unconstrained policy $h^*$ remains optimal on this subpopulation, i.e., one should still treat with $A = 1$ when $\tau(X) > 0$, and treat with $A = 0$ when $\tau(X) \leq 0$. In our setting where the intervention probability is limited, the CPB takes the role that the CATE does in the alternative setting where the proportion of active treatment (i.e., $A = 1$) is limited \citep{luedtke2016b}. In terms of statistical estimation, working with the CPB is more challenging due to added smoothness violations arising from thresholding $\beta$ at its $(1 - \delta)$-quantile, as we describe in Section~\ref{sec:estimation}.


From Proposition~\ref{prop:optimal}, we can also describe the gap between mean outcomes in the unconstrained setting (i.e., treating all subjects according to $h^*$) and optimal mean outcomes under our resource constraint. In practical settings, quantification of this gap may help in allocating resources and choosing an appropriate budget $\delta$.

\begin{corollary}\label{cor:gap}
    Under the assumptions of Proposition~\ref{prop:optimal},
    \[\mathbb{E}(Y(h^*)) - \mathbb{E}[Y(d(\Delta_{\delta}^*, h^*))] = \mathbb{E}((1 - \Delta_{\delta}^*)\beta) \leq (1 - \delta) q_{1 - \delta},\] for any $\delta \in [0,1]$.
\end{corollary}

In Section~\ref{sec:estimation} below, we will discuss efficient estimation of the CPB, $\beta$, the optimal contact rule, $\Delta_{\delta}^*$, and the value of the optimal regime, $\mathbb{E}[Y(d(\Delta_{\delta}^*, h^*))]$.

\subsection{A one-number summary measure of the potential benefit of targeting}\label{sec:AUPBC}
In addition to the quantities described above, we now propose a one-number summary measure that quantifies the potential benefits of targeting treatment selection to key participant subgroups. This summary is built on the so-called Qini curve \citep{radcliffe2007}, which plots the (estimated) value of an (estimated) optimal treatment strategy over a range of budgets. In our setting, the Qini curve plots $\mathbb{E}[Y(d(\Delta_{\delta}^*, h^*))]$ against $\delta \in [0,1]$. Our proposed summary measure is the area between this curve---identified under Assumptions~\ref{ass:consistency}--\ref{ass:NUC} by $\mathbb{E}(\Delta_{\delta}^* \beta + Y)$---and the straight line $\mathbb{E}(\delta\beta + Y)$, representing how well the optimal targeted treatment rule performs over a random treatment rule which assigns $h^*$ with fixed probability $\delta$ and otherwise the natural value $A$ occurs. We term this quantity the area under the potential benefit curve (AUPBC); the Qini curve and AUPBC in synthetic examples are illustrated in Figure~\ref{fig:Qini-scenarios}.

Given Assumptions~\ref{ass:consistency}--\ref{ass:NUC}, we can write the AUPBC in the following integral form:
\begin{equation}\label{eq:AUPBC}
    \mathcal{A} = \int_0^1 \mathbb{E}\left((\Delta_{\delta}^* - \delta)\beta\right) \, d \delta.
\end{equation}
Under some additional conditions, we can write the AUPBC in a more interpretable form.

\begin{proposition}\label{prop:AUPBC}
    Let $F_{\beta}(b) = \mathbb{P}[\beta(X) \leq b]$ be the cumulative distribution function of $\beta$. Under Assumptions~\ref{ass:consistency}--\ref{ass:NUC}, and assuming that $F_\beta$ is continuous and strictly increasing on the support of $\beta$, the AUPBC may be written as
    \[
    \mathcal{A} = \mathbb{E}\left(\beta\left\{F_{\beta}(\beta) - \frac{1}{2}\right\}\right) = \mathrm{Cov}(\beta, F_{\beta}(\beta)).
    \]
\end{proposition}

The form of the AUPBC  in Proposition~\ref{prop:AUPBC}
translates the geometric area quantity that we started with into a simple property of the distribution of the CPB. Moving forward, though, we will prefer the more general formula given by~\eqref{eq:AUPBC}, as it does not rely on the additional assumptions of Proposition~\ref{prop:AUPBC}.

Note that, since $\Delta_{\delta}^* \leq 1$, we have $\mathcal{A} \leq \frac{1}{2}\mathbb{E}(\beta)$, i.e, the AUPBC is bounded by the area of the region upper-left of the untargeted effect curve $\mathbb{E}(\delta \beta)$. In view of this fact, when $\mathbb{E}(\beta) > 0$ we can define a normalized AUPBC quantity,
\[\overline{\mathcal{A}} = 2\mathcal{A} / \mathbb{E}(\beta) = 2\frac{\int_{0}^1 \mathbb{E}(\Delta_{\delta}^* \beta) \, d\delta}{\mathbb{E}(\beta)} - 1,\]
so that $\overline{\mathcal{A}} \in [0,1]$. This allows for better comparison across contexts that differ in scale.

\subsection{Illustrative examples}
To provide some intuition for our framework and the functionals we have introduced, we briefly illustrate some hypothetical scenarios. In all scenarios, we consider the structural equations 
\[X \sim \mathrm{Unif}(-2, 2), \, A \mid X \sim \mathrm{Bernoulli}(\pi(X)), \, Y(a) = (a - \pi(X))\tau(X) + \epsilon.\] for some auxiliary mean-zero error $\epsilon$. Thus, Assumption~\ref{ass:NUC} holds, the CATE is given by $\tau$, the outcome has mean $0$, and the CPB is given by $\beta = \tau(\mathds{1}(\tau > 0) - \pi)$. Consider the following specifications for $(\pi, \tau)$:

\vspace{2mm}
\noindent
\textbf{Scenario 1}: $\pi(X) = \frac{1}{2}$, $\tau(X) = X$

\vspace{2mm}
\noindent
\textbf{Scenario $\mathbf{1^*}$}: $\pi(X) = 1 - \frac{|X|}{2}$, $\tau(X) = 1$

\vspace{2mm}
\noindent
\textbf{Scenario 2}: $\pi(X) = \frac{1}{2}$, $\tau(X) = \frac{3}{16} X^5$

\vspace{2mm}
\noindent
\textbf{Scenario $\mathbf{2^*}$}: $\pi(X) = \mathds{1}(X > 0)\left\{1 - \left(\frac{X}{2}\right)^4\right\} + \mathds{1}(X \leq 0)\left(\frac{X}{2}\right)^4$, $\tau(X) = 
\frac{3}{2}X$

\vspace{2mm}
These examples are designed to emphasize the roles that both treatment effects and selection behavior play in determining the CPB, and hence in optimally tailoring interventions. By construction, the CPB is $\beta(X) = \frac{1}{2}|X|$ in Scenarios 1 and $1^*$, while it is $\beta(X) = \frac{3}{32}|X|^5$ in Scenarios 2 and $2^*$.  In Scenarios 1 and 2, the CATE is moderately to highly explained by $X$, while treatment selection is completely unexplained by $X$. On the other hand, in Scenarios $1^*$ and $2^*$, the benefit in targeting interventions stems primarily from the relation between $X$ and treatment selection. Despite these differing mechanisms, Scenarios 1 and $1^*$ are interchangeable in terms of the effects of optimally targeting subgroups for intervention, and similarly for Scenarios 2 and $2^*$. 

Figure~\ref{fig:Qini-scenarios} plots the Qini curves for Scenario 1/$1^*$ and Scenario 2/$2^*$. While the optimal unconstrained value is constant across settings, $\mathbb{E}(Y(h^*)) = \mathbb{E}(\beta(X)) = \frac{1}{2}$, the shapes of the Qini curves demonstrate different impacts of tailored interventions under real-world budget constraints. Targeting is much more impactful under Scenario 2/$2^*$ (normalized AUPBC $\overline{\mathcal{A}} \approx 0.714$) than under Scenario 1/$1^*$ (normalized AUPBC $\overline{\mathcal{A}} = \frac{1}{3} $). Indeed, under Scenario 2/$2^*$, 80\% of the peak average outcome (compared to baseline) is achieved at $\delta \approx 0.25$, whereas under Scenario 1/$1^*$, $\delta \approx 0.55$ is required for the same performance.

\begin{figure}[t]
  \centering
  \includegraphics[width = 0.8\linewidth]{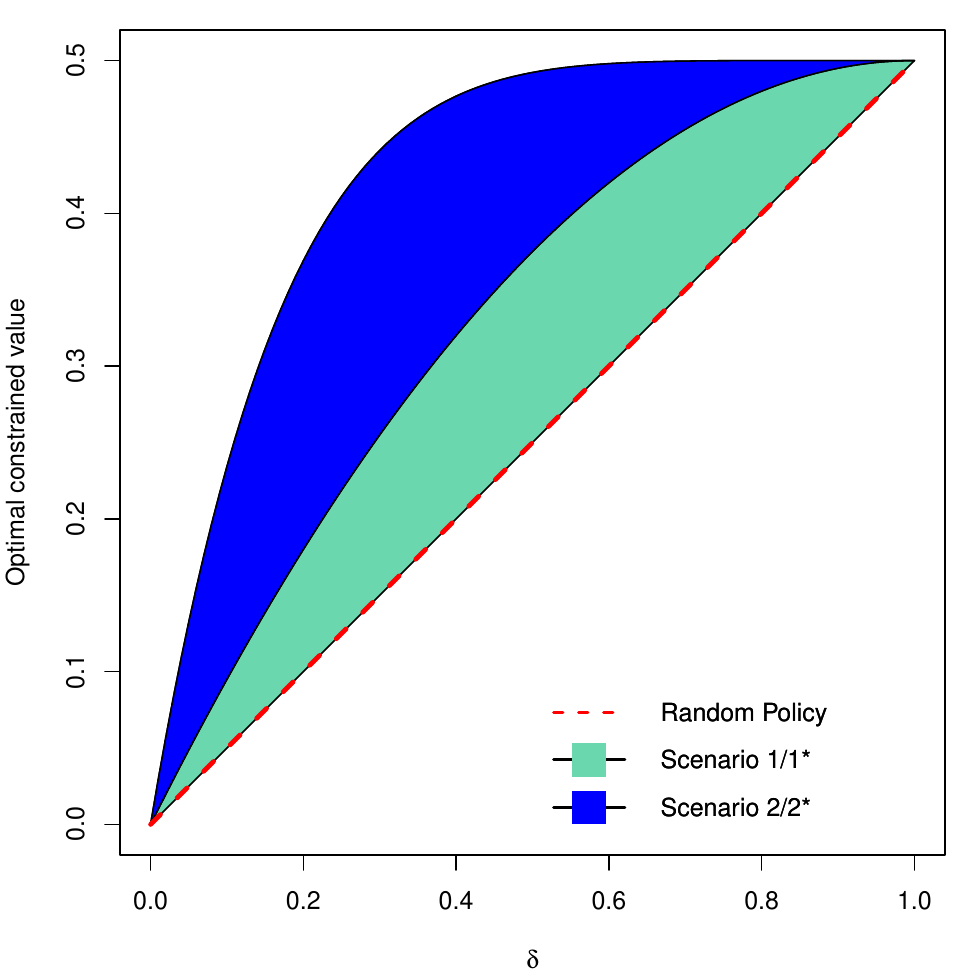}
  \caption{The Qini curves and AUPBC for two data generating scenarios. The optimal constrained value $\mathbb{E}(Y(d(\Delta_{\delta}^*, h^*)))$ is plotted against the budget parameter $\delta$, for Scenario 1/$1^*$ (boundary between green and blue regions) and 2/$2^*$ (boundary above blue region). The AUPBC for Scenario 1/$1^*$ is represented by the green shaded area, and for Scenario 2/$2^*$ is the sum of the green and blue shaded areas. The value of the random policy $d(\delta, h^*)$ is represented by the diagonal red dashed line.}
  \label{fig:Qini-scenarios}
\end{figure}

\section{Estimation}
\label{sec:estimation}
In this section, we develop efficient and robust estimators of the parameters laid out in Section~\ref{sec:estimands}.
We first focus on estimating the optimal  value of the optimal constrained treatment rule in Section~\ref{sec:estimation-value};
developing and analyzing this estimator highlights the sources of non-smoothness inherent to all of the functionals under study and serves as a building block for the remaining estimators.
In particular, we propose margin conditions that permit efficient estimation in the face of non-smoothness that we will use throughout our development.
We then build on this analysis in Sections~\ref{sec:estimation-CPB}, \ref{sec:estimation-rule}, \ref{sec:estimation-AUPBC}, to develop estimators of the CPB function $\beta$, the optimal contact rule $\Delta_{\delta}^*$, and the AUPBC summary measure $\mathcal{A}$, respectively.

As a matter of notational shorthand, we write $\mathbb{P}_n(f) = \frac{1}{n}\sum_{i=1}^n f(O_i)$ for the empirical mean of the function $f$ over the batch data, $\lVert f \rVert = \left\{\int f(o)^2 \, d\mathbb{P}(o)\right\}^{1/2}$ for the $L_2(\mathbb{P})$ norm of $f$, and $\lVert f \rVert_{\infty} = \mathrm{ess\,sup}\{f\}$ for the $L_{\infty}(\mathbb{P})$ norm of $f$. Throughout, we suppose that the propensity score and outcome models, $(\widehat{\pi}, \widehat{\mu}_0, \widehat{\mu}_1)$, are fit in a separate sample $D^n$, independent of the batch data $O_1, \ldots, O_n$. In practice, one can split a single sample and perform cross-fitting \citep{bickel1988, robins2008, zheng2010, chernozhukov2018}. Throughout we will assume a single split for simplicity, as analysis of procedures averaged across multiple splits follows in a straightforward fashion. We write $\mathbb{P}(f)= \int f(o) \, d\mathbb{P}(o)$, where $O \independent D^n$, for any (possibly data-dependent) function $f$.  For reference, notation for important functions and functionals are collected in Table~\ref{tab:nuisance}.

\subsection{Optimal value}\label{sec:estimation-value}
We begin by developing an estimator for the marginal quantity $V_{\delta} = \mathbb{E}\left(\Delta_{\delta}^*\beta + Y\right)$,
which represents, by Propositions~\ref{prop:ident} and~\ref{prop:optimal}, the value of the optimal constrained treatment rule of the form~\eqref{eq:contact-policy}, under Assumptions~\ref{ass:consistency}--\ref{ass:NUC}.
Our goal is to construct an efficient and robust estimator for $V_{\delta}$, by exploiting nonparametric efficiency theory and influence functions \citep{bkrw1993, tsiatis2007, kennedy2022review}. Unfortunately, this functional is not pathwise differentiable without further conditions. Specifically, non-smoothness arises in this setting for two reasons: (i) the presence of the indicator function $h^*$ in $\beta = \tau(h^* - \pi)$, and (ii) the indicator function $\Delta_{\delta}^*$ itself. The former source of non-smoothness arises also when estimating the value of the optimal unconstrained treatment regime, $\mathbb{E}(Y(h^*))$ \citep{luedtke2016a}, while the latter is unique to this constrained setting.

To overcome both sources of non-smoothness, we introduce two margin conditions that rule out concentration near the points of non-differentiability. These will permit faster rates of convergence, and weaker conditions under which our proposed estimator will achieve parametric convergence rates and nonparametric efficiency.

\begin{assumption}[Margin condition]\label{ass:margin}
    For some $a, b > 0$,
    \[\mathbb{P}\left[|\tau(X)| \leq t\right]
  \lesssim t^{a}, \text{ and }\mathbb{P}[|\beta(X) - q_{1 - \delta}
  | \leq t] \lesssim t^{b} \text{ for all } t \geq 0,\]
  where $u \lesssim v$ indicates $u \leq C v$ for some universal constant $C > 0$.
\end{assumption}

The conditions in Assumption~\ref{ass:margin} are similar to those invoked in the classification literature~\citep{tsybakov2004, audibert2007}, and more recently in a number of problems in causal inference and policy learning~\citep{qian2011, luedtke2016a, kennedy2020b, kallus_harm_2022, DAdamo_2023, levis2023, benmichael_asymm_2024}. Heuristically, Assumption~\ref{ass:margin} asserts that $\tau$ and $\beta$ do not concentrate too much near zero and $q_{1-\delta}$, respectively. This assumption is guaranteed to hold with $a = 1$ and $b = 1$, when the density of $\tau$ is bounded near zero, and the density of $\beta$ is bounded near the quantile $q_{1 - \delta}$. We expect these to be relatively mild conditions in most cases, though we note that for the first of these, we have to at least rule out the possibility of a point mass of $\tau$ at zero. In other words, the margin condition would fail if the CATE is exactly zero for a proportion of the population.

In view of the margin condition in Assumption~\ref{ass:margin}, we proceed by estimating the indicator functions $h^*$ and $\Delta_{\delta}^*$ with plug-in estimators, while targeting the smooth components of $V_{\delta}$ in a debiased manner using influence functions. Concretely, we define
\begin{equation}\label{eq:beta-pseudo}
    \phi(O; \mathbb{P}) = \left(h^* - \pi\right)\left\{\frac{A}{\pi} - \frac{1 - A}{1 - \pi}\right\}\left(Y - \mu_A\right) + \tau\left(h^* - A\right).
\end{equation}
The function $\phi$ satisfies $\mathbb{E}(\phi(O; \mathbb{P}) \mid X) = \tau(h^* - \pi) = \beta$, and has a crucial second-order bias property, as elaborated in the following result.

\begin{lemma}\label{lemma:pseudo-bias}
Let $\widetilde{\mathbb{P}}$ be an alternative fixed distribution on $O$. Then the function $\phi$ defined in~\eqref{eq:beta-pseudo} satisfies the following conditional bias decomposition:
\begin{align*}
    &\mathbb{E}\left(\phi(O; \widetilde{\mathbb{P}}) - \phi(O;\mathbb{P}) \mid X\right) \\
    &= \left(\widetilde{h}^* - \widetilde{\pi}\right) \sum_{a=0}^1
      \frac{\{\widetilde{\pi} - \pi\}
      \{\widetilde{\mu}_a - \mu_a\}}{a\widetilde{\pi} + (1 - a)(1 - \widetilde{\pi})} + \{\widetilde{\tau} - \tau\}\{\widetilde{\pi} - \pi\}
      + \{\widetilde{h}^* - h^*\}\tau,
\end{align*}
where $(\widetilde{\mu}_0, \widetilde{\mu}_1, \widetilde{\pi}, \widetilde{\tau}, \widetilde{h}^*)$ represent the corresponding nuisance functions under $\widetilde{\mathbb{P}}$.
\end{lemma}

The result of Lemma~\ref{lemma:pseudo-bias} motivates the use of $\phi$ as a robust pseudo-outcome for direct estimation of the CPB $\beta$---we discuss this in more detail in Section~\ref{sec:estimation-CPB}. Note that while the final term in the bias expression, $\{\widetilde{h}^* - h^*\}\tau$, would normally be first order, it is controlled under the margin condition in Assumption~\ref{ass:margin}.

Define the plug-in estimates of the CATE and the optimal treatment rule as $\widehat{\tau} = \widehat{\mu}_1 - \widehat{\mu}_0$, $\widehat{h}^* = \mathds{1}(\widehat{\tau} > 0)$. For a given CPB estimator $\widehat{\beta}$ constructed from $D^n$ (e.g., simply taking the plug-in $\widehat{\tau}\{\widehat{h}^* - \widehat{\pi}\}$, or else the estimator described in Section~\ref{sec:estimation-CPB}), let $\widehat{q}_{1 - \delta}$ solve $ \mathbb{P}_n\left[\widehat{\beta}(X) > \widehat{q}_{1 - \delta}\right] = \delta$ up to $o_{\mathbb{P}}(n^{-1/2})$ error, 
and define the plug-in optimal constrained contact rule as $\widehat{\Delta}_{\delta}^* = \mathds{1}\left(\widehat{\beta} > \widehat{q}_{1 - \delta}\right)$. The proposed estimator of $V_{\delta}$ is then given by
\[\widehat{V}_{\delta} = \mathbb{P}_n\left[\widehat{\Delta}_{\delta}^*\phi(O; \widehat{\mathbb{P}}) + Y\right],\]
where $\phi(O; \widehat{\mathbb{P}})$ is obtained from equation~\eqref{eq:beta-pseudo}, plugging in $(\widehat{\pi}, \widehat{\mu}_0, \widehat{\mu}_1)$ (and the derived nuisance estimates $\widehat{\tau}, \widehat{h}^*$).

Our first main result describes the rate of convergence of $\widehat{V}_{\delta}$ to $V_{\delta}$, and yields sufficient conditions under which asymptotic normality and nonparametric efficiency are guaranteed.

\begin{theorem}\label{thm:est-value}
  Assume that $\lVert \widehat{\mu}_0 -\mu_0\rVert + \lVert \widehat{\mu}_1 -\mu_1\rVert + \lVert \widehat{\pi} -\pi\rVert + \lVert \widehat{\beta} -\beta\rVert + |\widehat{q}_{1 - \delta} - q_{1 - \delta}| = o_{\mathbb{P}}(1)$. Moreover, assume that there exists $\epsilon > 0, M > 0$ such that $\mathbb{P}[\epsilon \leq \pi \leq 1 - \epsilon] = \mathbb{P}[\epsilon \leq \widehat{\pi} \leq 1 - \epsilon] = 1$, $\mathbb{P}[|Y| \leq M] = 1$. Then, under the margin condition (Assumption~\ref{ass:margin}),
  \[\widehat{V}_{\delta} - V_{\delta} =  O_{\mathbb{P}}\left(\frac{1}{\sqrt{n}} + R_{1,n} + R_{2,n} + R_{3,n}\right),\]
  where 
  \[R_{1,n} = \left\lVert \widehat{\pi} - \pi \right\rVert
    \left(\lVert \widehat{\mu}_0 -\mu_0\rVert + \lVert \widehat{\mu}_1
      -\mu_1\rVert\right),\]
  and
  \[R_{2,n} = \left\lVert \widehat{\tau} - \tau
      \right\rVert_{\infty}^{1 + a}, \ R_{3,n} = \left(\lVert \widehat{\beta} - \beta
      \rVert_{\infty} + |\widehat{q}_{1 - \delta} - q_{1 -
        \delta}|\right)^{1 + b}.\] 
  If, in addition, $R_{1,n} + R_{2,n} + R_{3,n} = o_{\mathbb{P}}(n^{-1/2})$, then \[\sqrt{n}(\widehat{V}_{\delta} - V_{\delta}) \overset{d}{\to} \mathcal{N}\left(0, \sigma(\delta)^2\right),\]
  where $\sigma(\delta)^2 = \mathrm{Var}(\Delta_{\delta}^* \left\{\phi(O;\mathbb{P}) - q_{1-\delta}\right\} + Y)$.
  
\end{theorem}

By Theorem~\ref{thm:est-value}, the error in estimation of the optimal constrained value $V_{\delta}$ consists of (i) a product bias (``doubly robust'') term $R_{1,n}$, which will be small if either $\pi$ or $\mu_a$ are estimated well; (ii) a CATE-based error term $R_{2,n}$, which will be small if $\tau$ is estimated well and/or the margin coefficient $a$ is larger; and (iii) a CPB-based error term $R_{3,n}$, which will be small if $\beta$ is estimated well and/or $b$ is larger. While $R_{1,n}$ is the usual doubly robust bias that arises in average treatment effect estimation, and $R_{2,n}$ appears in the estimation of the optimal unconstrained value $\mathbb{E}(Y(h^*))$ \citep{luedtke2016a}, the third bias term $R_{3,n}$ is new to this problem and arises due to non-smoothness from thresholding $\beta$ at $q_{1 - \delta}$. In the case that the three bias terms are small enough to yield asymptotic normality of $\widehat{V}_{\delta}$, one can obtain asymptotically valid inference via simple Wald-based confidence intervals:
  $\widehat{V}_{\delta} \pm z_{1- \alpha/2}\frac{\widehat{\sigma}(\delta)}{\sqrt{n}}$, where
  \[\widehat{\sigma}(\delta)^2 = \mathbb{P}_n\left[\left(\widehat{\Delta}_{\delta}^* \left\{\phi(O;\widehat{\mathbb{P}}) - \widehat{q}_{1-\delta}\right\} + Y - \left\{\widehat{V}_{\delta} - \delta \widehat{q}_{1 - \delta}\right\}\right)^2\right],\]
  where $z_{1 - \alpha/2}$ is the $(1 - \alpha/2)$-quantile of the standard normal distribution.

\subsection{Conditional potential benefit}\label{sec:estimation-CPB}
To estimate the CPB function $\beta$, we will again proceed nonparametrically, and aim to construct an estimator that can perform well in a variety of scenarios. In particular, we will propose a model agnostic, robust, two-stage pseudo-outcome regression-based estimator, similar to several doubly robust ``DR-Learners'' that have been proposed in other problems \citep{foster2023, kennedy2023}. For the theory to work well with the margin condition (Assumption~\ref{ass:margin}), we use the $L_2(\mathbb{P})$ oracle inequality for pseudo-outcome regression developed in~\citet{rambachan2022}. The relevant theory is reviewed in Appendix~\ref{app:L2}.

To motivate the proposed DR-Learner, recall that under Assumptions~\ref{ass:consistency}--\ref{ass:NUC}, 
\[\beta(X) = \mathbb{E}(Y(h^*) - Y \mid X) = \mathbb{E}(\{h^* - A\}\{Y(1) - Y(0)\} \mid X).\] If we could observe the counterfactuals $Y(0), Y(1)$, and know $h^*$ without error, then one could simply regress $(h^* - A)(Y(1) - Y(0))$ on $X$ to obtain estimates of $\beta(X)$. Of course, we cannot observe both levels of the counterfactual in practice. Alternatively, if the propensity score $\pi$ and the optimal unconstrained policy $h^*$ were known, then one could consider the inverse probability-weighted learner, regressing $(h^* - \pi)\frac{A - \pi}{\pi(1 - \pi)}Y$ on $X$. This would (for most regression procedures) have error on the same order as the oracle procedure that had knowledge of $Y(0), Y(1)$, since this pseudo-outcome also has conditional mean exactly equal to $\beta$. Since in a typical observational study, neither $\pi$ nor $h^*$ would be known, we desire an alternative approach that can mimic the performance of these oracle methods under general nonparametric conditions.

In order to achieve convergence rates as close as possible to an oracle procedure, while avoiding specific structural assumptions, we propose a generic two-stage regression procedure based on the pseudo-outcome $\phi$ defined in~\eqref{eq:beta-pseudo}. This particular pseudo-outcome is chosen for its favorable influence function-like properties, namely that described in Lemma~\ref{lemma:pseudo-bias}. In what follows, $\widehat{\mathbb{E}}_n(f(O) \mid X = x)$ will denote a generic regression procedure $\widehat{\mathbb{E}}_n$ applied to the function $f$ on covariates $X$, evaluated at $x \in \mathcal{X}$.

\begin{algorithm}\label{alg:DR-CPB}
    Let $D^n = (O_{01}, \ldots, O_{0n})$ and $O^n = (O_1, \ldots, O_n)$ denote independent training and batch samples, respectively.
    \begin{enumerate}
        \item[Step (i)]
        Fit propensity score model $\widehat{\pi}$ and outcome regression models $(\widehat{\mu}_0, \widehat{\mu}_1)$ using $D^n$. Based on these, define $\widehat{\tau} = \widehat{\mu}_1 - \widehat{\mu}_0$, $\widehat{h}^* = \mathds{1}(\widehat{\tau} > 0)$.
        \item[Step (ii)]
        Compute the pseudo-outcome
        \[\widehat{\phi}(O) = \left(\widehat{h}^*(X) - \widehat{\pi}(X)\right)\left\{\frac{A}{\widehat{\pi}(X)} - \frac{1 - A}{1 - \widehat{\pi}(X)}\right\}\left(Y - \widehat{\mu}_A(X)\right) + \widehat{\tau}(X)\left(\widehat{h}^*(X) - A\right),\]
        and regress $\widehat{\phi}$ on $X$ in $O^n$, to give
        $\widehat{\beta}_{\mathrm{dr}} = \widehat{\mathbb{E}}_n\left(\widehat{\phi}(O) \mid X = x\right)$,
        for each $x \in \mathcal{X}$.
        \item[Step (iii)] Optionally, swap the roles of $D^n$ and $O^n$ in Steps (i) and (ii) above, then average the resulting estimates.
    \end{enumerate}
\end{algorithm}

Our next main result provides error bounds for the DR-Learner $\widehat{\beta}_{\mathrm{dr}}$ described in Algorithm~\ref{alg:DR-CPB}, relative to a procedure that regresses the oracle pseudo-outcome $\phi(O; \mathbb{P})$ on $X$---again, since the conditional mean of $\phi(O; \mathbb{P})$ is the true $\beta$, this approach would typically yield error on the same order as regressing $(h^* - A)(Y(1) - Y(0))$ on $X$.
\begin{theorem}\label{thm:est-CPB}
    Assume that the second stage regression procedure $\widehat{\mathbb{E}}_n$ is $L_2(\mathbb{P})$-stable with respect to the metric $\rho$ (see Appendix~\ref{app:L2} for precise definitions), and that $\rho\left(\widehat{\phi}, \phi(\cdot ; \mathbb{P})\right) \overset{\mathbb{P}}{\to} 0$. Define $\widetilde{\beta}(x) = \widehat{\mathbb{E}}_n\left(\phi(O; \mathbb{P}) \mid X = x\right)$ to be an oracle estimator regressing the true pseudo-outcome on covariates $X$ in $O^n$, and let $R_n^* = \mathbb{E}\left(\lVert \widetilde{\beta} - \beta\rVert\right)$ be the oracle $L_2(\mathbb{P})$ risk. 
    Then
    \[\lVert \widehat{\beta}_{\mathrm{dr}} - \widetilde{\beta}\rVert = \left\lVert \widehat{\mathbb{E}}_n\left(\widehat{b}(X) \mid X = \, \cdot \, \right)\right\rVert + o_{\mathbb{P}}(R_n^*),\]
    where
    \begin{equation}\label{eq:bias-CPB}
    \widehat{b} = \left(\widehat{h}^* - \widehat{\pi}\right) \sum_{a=0}^1
      \frac{\{\widehat{\pi} - \pi\}
      \{\widehat{\mu}_a - \mu_a\}}{a\widehat{\pi} + (1 - a)(1 - \widehat{\pi})} + \{\widehat{\tau} - \tau\}\{\widehat{\pi} - \pi\}
      + \{\widehat{h}^* - h^*\}\tau.
    \end{equation}
    Further, if $\left\lVert \widehat{\mathbb{E}}_n\left(\widehat{b}(X) \mid X = \, \cdot \, \right)\right\rVert = o_{\mathbb{P}}(R_n^*)$, then $\lVert \widehat{\beta}_{\mathrm{dr}} - \beta\rVert = \lVert \widetilde{\beta} - \beta\rVert + o_{\mathbb{P}}(R_n^*)$, i.e., $\widehat{\beta}_{\mathrm{dr}}$ achieves the oracle rate of convergence.
\end{theorem}

The result of Theorem~\ref{thm:est-CPB} is quite general, and the conditions are relatively mild. We review in Appendix~\ref{app:L2} that the stability requirement is satisfied when $\widehat{\mathbb{E}}_n$ is a linear smoother, and $\rho$-consistency of the estimated pseudo-outcome $\widehat{\phi}$ will typically follow from pointwise or $L_2(\mathbb{P})$ consistency of $\widehat{\phi}$ itself. The oracle risk $R_n^*$ is well-characterized for many estimators and families of structural assumptions. For instance, when $\beta$ is $s$-smooth in the H{\"o}lder sense, and $\widehat{\mathbb{E}}_n$ is an appropriately tuned local polynomial or series estimator, then $R_n^* \asymp n^{- \frac{1}{2 + d/s}}$ (see \citet{gyorfi2002, tsybakov2009} for definitions and results). For the bias term, the argument in the proof of Theorem~\ref{thm:est-value} establishes that \[\big\lVert \widehat{b}\big\rVert = O_{\mathbb{P}}\left(\max_{a \in \{0,1\}}\lVert \{\widehat{\mu}_a - \mu_a\}\{\widehat{\pi} - \pi\}\rVert + \lVert \widehat{\tau} - \tau\rVert_{\infty}^{1 + a / 2}\right).\]
The first of these bias terms can be further bounded using H{\"o}lder's inequality, and will be small when error in estimating $\mu_a$ or $\pi$ is small, whereas the second term depends on the error in estimating $\tau$ and the margin parameter from Assumption~\ref{ass:margin}. Since $\widehat{\mathbb{E}}_n\left(\widehat{b}(X) \mid X = x\right)$ is an estimate of the mean of $\widehat{b}$ at a given $x$, it may be thought of as a smoothed bias; in Appendix~\ref{app:L2}, we review that when $\widehat{\mathbb{E}}_n$ is a linear smoother, this smoothed bias can be expressed in terms of certain weighted norms of the bias $\widehat{b}$ itself.

\begin{remark}\label{rem:CPB-subs}
    While we present results for the fully conditional CPB, $\beta(X) = \mathbb{E}(Y(h^*) - Y \mid X)$, we note that these extend in a straightforward way to estimation of coarser CPB metrics, e.g., $\mathbb{E}(Y(h^*) - Y \mid V) \equiv \mathbb{E}(\beta \mid V)$ for a given subset of covariates $V \subseteq X$. Such an estimand may be of practical interest if one wishes to summarize the potential benefit of a targeted intervention in simpler terms, say across the range of one continuous baseline covariate. Procedurally, one simply changes step (ii) in Algorithm~\ref{alg:DR-CPB} to regress the same pseudo-outcome on $V$ instead of $X$. An analog of Theorem~\ref{thm:est-CPB} follows by the same arguments in Appendix~\ref{app:L2}, with the notable difference that the bias term $\widehat{b}$ in~\eqref{eq:bias-CPB} is replaced by $\mathbb{E}(\widehat{b} \mid V)$. Moreover, since the $V$-conditional CPB, $\mathbb{E}(\beta \mid V)$, involves extra marginalization, it is lower-dimensional than $\beta$ (and possibly less complex), so  the oracle risk $R_n^*$ in Theorem~\ref{thm:est-CPB} may be of smaller order, i.e., converging faster to zero.
\end{remark}

\subsection{Optimal rule}\label{sec:estimation-rule}
We focus here on estimating the optimal contact rule $\Delta_{\delta}^* = \mathds{1}(\beta > q_{1-\delta})$ itself. In estimation of its value in Section~\ref{sec:estimation-value}, we took an arbitrary estimator of $\widehat{\beta}$ based on training data $D^n$, computed the in-sample approximate quantile $\widehat{q}_{1 - \delta}$ that solved $\mathbb{P}_n\left[\widehat{\beta}(X) > \widehat{q}_{1 - \delta}\right] = \delta + o_{\mathbb{P}}(n^{-1/2})$, and worked with $\widehat{\Delta}_{\delta}^* = \mathds{1}\left(\widehat{\beta} > \widehat{q}_{1 - \delta}\right)$. We will see that this estimated contact rule---together with the other plug-in estimator $\widehat{h}^* = \mathds{1}(\widehat{\tau} > 0)$ for the optimal unconstrained policy---performs well with respect to a natural error metric.

In the classification literature, one often characterizes error of an estimated classifier by quantifying ``regret'' relative to a true optimal classifier.
In our setting, as the treatment rules under study are ordered with respect to their value, we can treat value relative to the optimal rule as a regret function: in view of Proposition~\ref{prop:ident}, let $\mathcal{R}_{\delta}(\Delta, h) = \mathbb{P}(\tau\{\Delta_{\delta}^* (h^* - \pi) - \Delta (h - \pi)\})$, for any $\Delta, h$ possibly dependent on $D^n$. The next result quantifies the error of the estimated pair $(\widehat{\Delta}_{\delta}^*, \widehat{h}^*)$ relative to the true optimal pair $(\Delta_{\delta}^*, h^*)$.

\begin{proposition}\label{prop:bias-rule}
    Under Assumptions~\ref{ass:consistency}--\ref{ass:NUC},
    \[\mathbb{E}[Y(d(\Delta_{\delta}^*, h^*) - Y(d(\widehat{\Delta}_{\delta}^*, \widehat{h}^*)) \mid X, D^n] = (\Delta_{\delta}^* - \widehat{\Delta}_{\delta}^*)\beta + \widehat{\Delta}_{\delta}^* \tau (h^* - \widehat{h}^*).\]
    Moreover, if Assumption~\ref{ass:margin} holds, then
    \[\mathcal{R}_{\delta}\left(\widehat{\Delta}_{\delta}^*, \widehat{h}^* \right) = O_{\mathbb{P}}\left(\frac{1}{\sqrt{n}} + R_{2,n} + R_{3,n}\right),\]
    where $R_{2,n}$ and $R_{3,n}$ are as defined in Theorem~\ref{thm:est-value}.
\end{proposition}

Similar to the estimated value in Section~\ref{sec:estimation-value}, the regret depends on how well $\tau$ and $\beta$ are estimated, through $R_{2,n}$ and $R_{3,n}$, respectively. We reiterate that this result is agnostic to the choice of estimators $\widehat{\tau}$ and $\widehat{\beta}$, as long as the plug-in procedure for $\widehat{h}^*$ and $\widehat{\Delta}_{\delta}^*$ described above is followed.

This result is important from a practical perspective when considering implementation of the estimated policies. That is, the result gives a bound on how well the estimated policy would perform (on average) compared to the true optimal policy, since regret is defined in terms of value, i.e., the expected outcome under different treatment rules. With larger sample sizes and well estimated nuisance functions, one can expect closer to optimal performance of the estimated treatment rule.

\begin{remark}
    While Proposition~\ref{prop:bias-rule} quantifies the regret gap for $(\Delta_{\delta}^*, h^*)$ jointly, one might be interested in error of the optimal contact rule $\Delta_{\delta}^*$ on its own. This can be done by defining an alternative regret $\widetilde{\mathcal{R}}_{\delta}(\Delta) = \mathbb{P}(\tau(h^* - \pi)\{\Delta - \Delta_{\delta}^*\})$, where $h^*$ is fixed at the truth. In this case, one can show $\widetilde{\mathcal{R}}_{\delta}(\widehat{\Delta}_{\delta}^*) = O_{\mathbb{P}}(1/\sqrt{n} + R_{3,n})$. Similarly, one can isolate the error in estimating $h^*$, either by fixing $\Delta$ in the regret definition or considering mean counterfactuals when implementing $\widehat{h}^*$ on the whole population---in either case, the regret bound can be shown to be of the order $O_{\mathbb{P}}(R_{2,n})$.
\end{remark}

\subsection{Area under the potential benefit curve}\label{sec:estimation-AUPBC}
Our last focus of estimation will be the AUPBC summary measure introduced in Section~\ref{sec:AUPBC}, $\mathcal{A} = \int_0^1 \mathbb{E}((\Delta_\delta^* - \delta)\beta) \, d\delta$, as well as its normalized version $\overline{\mathcal{A}} = 2\frac{\mathcal{A}}{\mathbb{E}(\beta)}$. By the same motivation as the proposed estimator for the optimal constrained value at a point, we propose the following estimator for the AUPBC:
\[\widehat{\mathcal{A}} = \int_0^1 \mathbb{P}_n\left[\phi(O; \widehat{\mathbb{P}}) \left(\widehat{\Delta}_{\delta}^* - \delta\right)\right] \, d\delta = \int_{0}^1 \widehat{V}_{\delta} \, d\delta - \mathbb{P}_n[Y] - \frac{1}{2}\mathbb{P}_n\left[\phi(O; \widehat{\mathbb{P}})\right],\]
where $\phi(O; \widehat{\mathbb{P}})$ and $\widehat{\Delta}_{\delta}^*$ are as defined in Section~\ref{sec:estimation-value}. Similarly, the normalized AUPBC can be estimated via
\[\widehat{\overline{\mathcal{A}}} = 2\frac{\int_0^1 \mathbb{P}_n\left[\phi(O; \widehat{\mathbb{P}}) \widehat{\Delta}_{\delta}^*\right] \, d\delta}{\mathbb{P}_n\left[\phi(O; \widehat{\mathbb{P}})\right]} - 1.\]
Note that one can use any available numerical tools to compute the necessary integrals up to arbitrary precision.

Just as in the pointwise value case, the proposed estimators are debiased through the use of influence functions, and taking advantage of the margin condition in Assumption~\ref{ass:margin}. The following result gives the asymptotic behavior of $\widehat{\mathcal{A}}$ and $\widehat{\overline{\mathcal{A}}}$.

\begin{theorem}\label{thm:est-AUPBC}
    Suppose the assumptions of Theorem~\ref{thm:est-value} hold, taking Assumption~\ref{ass:margin} to hold for all $\delta \in [0,1]$ with a common margin coefficient $b$. Further, assume that $\mathbb{P}[|\widehat{\mu}_a|\leq M] = 1$ for $a \in \{0,1\}$, $\sup_{\delta \in [0,1]} \left|\widehat{\Delta}_{\delta}^* \{\phi(O; \widehat{\mathbb{P}}) - q_{1 - \delta}\} - \Delta_{\delta}^* \{\phi(O; \mathbb{P}) - q_{1 - \delta}\}\right| = o_{\mathbb{P}}(1)$, and that $\sup_{\delta \in [0,1]}\left|\mathbb{P}_n[\widehat{\Delta}_{\delta}^*] - \delta\right| = o_{\mathbb{P}}(n^{-1/2})$. Then
    \[\widehat{\mathcal{A}} - \mathcal{A} = O_{\mathbb{P}}\left(\frac{1}{\sqrt{n}} + R_{1,n} + R_{2,n} + R_{3,n}^*\right),\]
    where $R_{1,n}$ and $R_{2,n}$ are as defined in Theorem~\ref{thm:est-value}, and
    \[R_{3,n}^* = \bigg(\lVert \widehat{\beta} - \beta\rVert_{\infty} + \sup_{\delta \in [0,1]}|\widehat{q}_{1 - \delta} - q_{1 - \delta}| \bigg)^{1 + b}.\]
    Similarly, under these same conditions, $\widehat{\overline{\mathcal{A}}} - \overline{\mathcal{A}} = O_{\mathbb{P}}\left(\frac{1}{\sqrt{n}} + R_{1,n} + R_{2,n} + R_{3,n}^*\right)$. In the case that $R_{1,n} + R_{2,n} + R_{3,n}^* = o_{\mathbb{P}}(n^{-1/2})$, we have
    \[\sqrt{n}\left(\widehat{\mathcal{A}} - \mathcal{A}\right) \overset{d}{\to} \mathcal{N}(0, \kappa^2) \text{ and }  \sqrt{n}\left(\widehat{\overline{\mathcal{A}}} - \overline{\mathcal{A}}\right) \overset{d}{\to} \mathcal{N}(0, \zeta^2),\]
    for asymptotic variances $\kappa^2 = \mathrm{Var}\left(N\right)$, $\zeta^2 = \mathrm{Var}\left(\frac{1}{\mathbb{E}(\beta)}\left\{2N - \overline{\mathcal{A}}\phi(O;\mathbb{P})\right\}\right)$, where
    \[N = \int_0^1 \left\{\Delta_{\delta}^*\phi(O;\mathbb{P}) - q_{1 - \delta}(\Delta_{\delta}^* - \delta)\right\} \, d\delta - \frac{1}{2}\phi(O; \mathbb{P}) .\]
\end{theorem}

We note that the extra conditions of Theorem~\ref{thm:est-AUPBC} (compared to Theorem~\ref{thm:est-value}) are relatively mild, and enable uniform convergence for $\widehat{V}_{\delta}$ over $\delta \in [0,1]$, which aids in the analysis of $\widehat{\mathcal{A}}$ and $\widehat{\overline{\mathcal{A}}}$. The other comments that followed Theorem~\ref{thm:est-value} apply here as well, for interpreting the bias terms $R_{1,n}, R_{2,n}, R_{3,n}^*$. Moreover, under the conditions of Theorem~\ref{thm:est-AUPBC} that yield asymptotic normality, one can construct simple asymptotically valid Wald-based confidence intervals using plug-in estimators of $\kappa^2$ and $\zeta^2$.

We remark that the functionals $\mathcal{A}$ and $\overline{\mathcal{A}}$ present challenging statistical estimation problems, and it seems that authors targeting analogous area measures in dynamic treatment regime settings tend to use simple plug-in estimators (e.g., \citet{imai2023}). 
In the same vein as \citet{ledell2015}, who estimate the area under the ROC curve in a binary classification setting, the proposed robust estimator of the AUPBC appeals to nonparametric efficiency theory, and the proposed methodology may be useful for constructing estimators of similar area measures in other dynamic treatment regime contexts.

\section{Data Application}
\label{sec:applications}

In this section, we illustrate the empirical application of our methodology on data from the (SPOT)light study, a prospective cohort of ward patients with deteriorating health \citep{harris2018, keele2019}. Here, we follow \citet{keele2019} and \citet{mcclean2024}, and use data on $n = $ 13,011 patients across the UK referred to critical care between November 1, 2010, and December 31, 2011, to study the impact of intensive care unit (ICU) transfer of critically ill patients on mortality. In addition to estimating the CPB to assess which subgroups are suboptimally allocated or not to the ICU, we are interested in evaluating treatment rules of the form~\eqref{eq:contact-policy}, in which decision makers (e.g., hospital staff) could select a subpopulation on which to intervene and provide tailored recommendations for ICU transfer and admission, in place of the status quo procedure. In a hospital setting, resource constraints may only allow for such interventions to be performed on a relatively small subset of patients, and there may be only a limited amount of information available to target this intervention.

We analyze a publicly available version of the (SPOT)light dataset, in which the rows are $n$ samples with replacement from the original data. Data and code to reproduce all analyses are available at~\href{https://github.com/alexlevis/CPB}{https://github.com/alexlevis/CPB}.

\subsection{Data and Models}

The exposure, $A$, we took to be transfer and admission to the ICU, while the outcome of interest, $Y$, was a binary indicator of survival 28 days after the admission decision. The (SPOT)light study involved measurement of a rich set of potential baseline confounders. In $X$, we included age, sex, diagnosis of sepsis, peri-arrest status, a measure of bed availability in the ICU, indicators for whether the visit took place during the weekend, or during winter, health care site, and three risk severity scores constructed based on physiology measurements: the Intensive Care National Audit \& Research  Centre physiology score \citep{harrison2007}, the Sepsis-related Organ Failure Assessment (SOFA) score \citep{vincent1996}, and the NHS National Early Warning Score \citep{london2012}. Based on the UK Critical Care Minimum Dataset levels of care \citep{danbury2015}, we also included in $X$ the level of care for the patient prior to assessment, as well as the recommended level of care after assessment.

To estimate $(\pi, \mu_0, \mu_1)$, we used Super Learner ensembles \citep{vanderlaan2007} of logistic generalized linear models, random forests, and regression trees. To estimate the CPB, we used the robust learner developed in Section~\ref{sec:estimation-CPB}, and fit a regression of $\phi(O; \widehat{\mathbb{P}})$ on $X$; in our results below, we summarize CPB across two key covariates by estimating $\mathbb{E}(Y(h^*) - Y \mid V)$ where $V$ is age or SOFA risk score (see Remark~\ref{rem:CPB-subs}), using as the second stage regression procedure (i.e., Step (ii) in Algorithm~\ref{alg:DR-CPB}) a Super Learner ensemble of a linear model, regression tree, a polynomial adaptive spline and a smoothing spline. We also estimated the value of optimal constrained policies as outlined in Section~\ref{sec:estimation-value}, across a grid of budget values $\delta \in [0,1]$, and combined these to estimate the normalized AUPBC summary measure for this setting. Since $V_{\delta}$ is monotone in $\delta$, we used the rearrangement procedure in~\citet{chernozhukov2009} to preserve monotonicity in the estimated values (this had a minor effect as estimates were nearly monotone without any rearrangement). For the CPB estimators, we used a single sample split, while for estimators of all other functionals, four-fold cross-fitting was employed.

Finally, using the methodology developed in Appendix~\ref{app:subset}, we consider the case where a decision maker might only have access to the variables \[W = \{\texttt{age}, \texttt{sex}, \texttt{sepsis}, \texttt{weekend}, \texttt{site}, \texttt{season}\} \subsetneq X.\]
For instance, in a clinical setting, more in-depth information such as the risk severity scores may not be readily available to clinicians or hospital staff. We estimate analogous optimal constrained rules where the contact rule and/or the ensuing policy on the selected subpopulation are restricted to only depend on $W$.

\subsection{Results}
The results for CPB estimation are illustrated in Figure~\ref{fig:CPB-ICU}. We can see that both age and SOFA score are predictive of potential benefit. For instance, we estimate that those above age $\sim\!\!55$ years, or above a SOFA score of $\sim\!\!2$, can expect $\geq 5$\% increased probability of 28-day survival under optimal treatment allocation compared to the status quo. Additionally, we estimate that increased age or SOFA scores are associated with even further gains from targeted intervention.

\begin{figure}[ht]
  \centering
  \includegraphics[width = 0.9\linewidth]{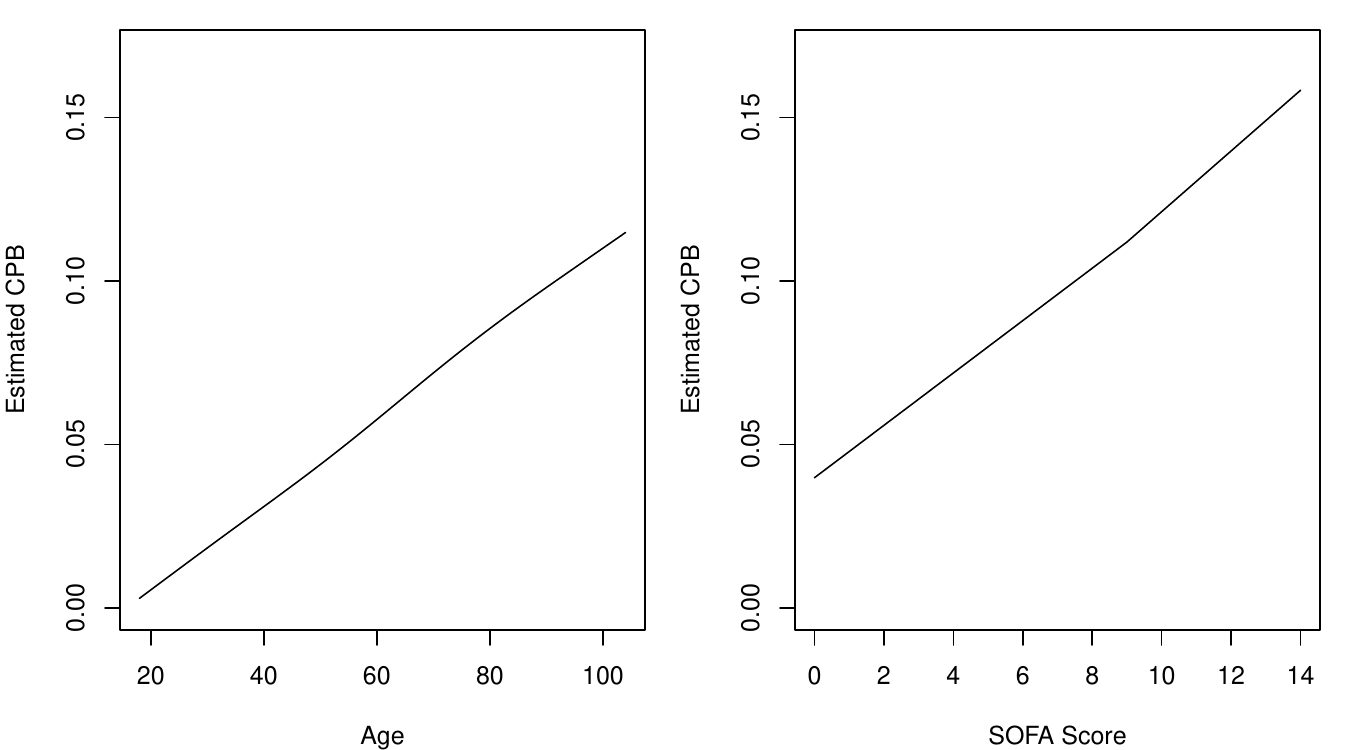}
  \caption{Estimated CPB across age and SOFA risk score in (SPOT)light data.}
  \label{fig:CPB-ICU}
\end{figure}

Next, in Figure~\ref{fig:Qini-ICU}, we plot the estimated Qini curves for the value of optimal constrained policies across all possible budgets. The estimated normalized AUPBC for unrestricted treatment rules (black curve) is $0.79$ (95\% CI: $0.73, 0.84$), but when we restrict contact rules to depend only on $W$ (blue curve), the AUPBC is $0.44$ ($0.39, 0.50$); this clearly demonstrates the importance of including key predictors of potential benefit to maximize outcomes under these treatment rules. Moreover, while the estimated value of the optimal unconstrained policy, $\widehat{V}_1 = \widehat{\mathbb{E}}(Y(h^*))$, is 0.84 ($0.83, 0.85$), when both the contact rule and ensuing policy are restricted to depend only on $W$ (red curve), the optimal unconstrained policy has value $0.80$ ($0.79, 0.81$). This is quite a substantial difference---especially relative to the mean survival probability under the status quo regime, $0.76$ ($0.75, 0.77$)---reflecting the impact of using maximally predictive covariates not only on the selection of subjects on whom to intervene, but also on the effectiveness of tailored treatment on the selected subpopulation.


\begin{figure}[ht]
  \centering
  \includegraphics[width = 0.8\linewidth]{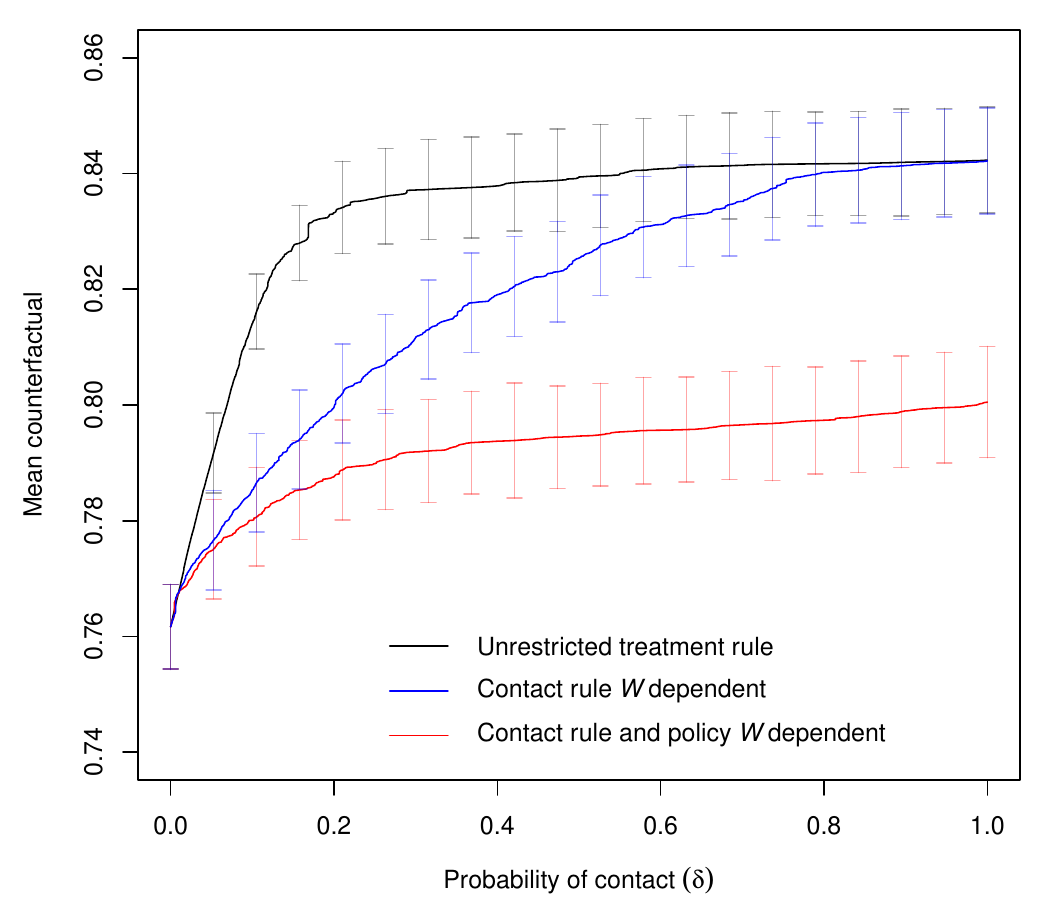}
  \caption{Estimated Qini curves for the (SPOT)light data. The black curve plots the estimated optimal constrained value $\mathbb{E}(Y(d(\Delta_{\delta}^*, h^*)))$---the expected survival under the optimal constrained policies---against the budget parameter $\delta$. The blue curve represents the estimated optimal constrained values when the contact rule is restricted to only depend on $W$, and the red curve when both the contact rule and ensuing policy may only depend on $W$. Vertical bars represent Wald-based pointwise 95\% confidence intervals.}
  \label{fig:Qini-ICU}
\end{figure}

\section{Discussion}
\label{sec:discussion}
In this work, we proposed the \textit{conditional potential benefit} (CPB), a general measure of the potential benefit of a targeted intervention compared to status quo. On the basis of the proposed CPB metric, we developed a framework for designing optimal intervention strategies under real-world constraints. Namely, we introduced contact rules, characterized optimal interventions that can tailor treatment on only a given proportion of the population, proposed an area under the potential benefit curve (AUPBC) summary measure, and developed efficient estimators of these treatment rules, their values, and other related quantities. These functionals can be estimated in any unconfounded observational study, and can provide valuable insight whether or not policy design is of primary interest. Such analysis can help in understanding the inefficiencies in given scientific contexts and suggest possible strategies for improvement, in that one can probe whether key subgroups with large treatment effects and/or high probability of incorrect treatment selection are being overlooked.

Several interesting extensions of our framework are possible, and will be followed up on in future work. First, one generalization of our constrained policy framework would involve variable costs that depend on covariates and, possibly, treatment options. For instance, it may be more costly to contact and induce treatment uptake for some subjects compared to others. In such cases, one could consider optimal treatment rules under a fixed total expected cost. A second issue is that of compliance, which is important in the application of dynamic treatment regimes more broadly. In particular, it will be important to extend our framework to account for possible imperfect compliance when implementing a tailored intervention. A third interesting avenue will be to explore longitudinal settings, and characterize constraints and optimal treatment rules that are defined across a series of time points. Lastly, it may be that the batch data to which the analyst has access differs in some way from the population on which an intervention is to be implemented. To deal with this, tools from the transportability and generalization literature will likely be useful \citep{cole2010, dahabreh2020, zeng2023}.

\section*{Acknowledgements}
EHK was supported by NSF CAREER
Award 2047444.

\vspace{10mm}

\begin{table}[ht]
\centering
\begin{tabular}{|c|c|c|}
\hline
Symbol & Definition & Description   \\
\hline
\hline
    $\pi$ & $\mathbb{P}[A = 1\mid X]$ & Propensity score  \\
     $\mu_a$  & $\mathbb{E}(Y \mid X, A = a)$ & Outcome model  \\
     \hline\hline
     $\tau$  & $\mu_1 - \mu_0$ & Conditional average treatment effect (CATE) \\
      $h^*$ & $\mathds{1}(\tau > 0)$ & Optimal unconstrained treatment rule\\
      $c$ & $h^*(1 - \pi) + (1 - h^*)\pi$ & Probability of suboptimal treatment\\
      $\beta$ & $\tau(h^* - \pi)$ & Conditional potential benefit (CPB) \\
      $\Delta_{\delta}^*$ & $\mathds{1}(\beta > q_{1 - \delta})$ & Optimal constrained contact rule\\
       \hline
       \hline
       $q_{1 - \delta}$ & solves $\delta = \mathbb{P}[\beta > q_{1 - \delta}]$& $(1 - \delta)$-quantile of CPB \\
       $V_{\delta}$ & $\mathbb{E}(\Delta_{\delta}^* \beta + Y)$ & Optimal constrained value \\
       $\mathcal{A}$ & $\int_0^1 \mathbb{E}\left(\{\Delta_{\delta}^* - \delta\}\beta\right)\, d\delta$ & Area under the potential benefit curve (AUPBC) \\
       $\overline{\mathcal{A}}$ & $2 \mathcal{A} / \mathbb{E}(\beta)$ & Normalized AUPBC \\
       \hline
\end{tabular}
\caption{Reference list for important functions and functionals}
\label{tab:nuisance}
\end{table}

\clearpage

\section*{References}
\vspace{-1cm}
\bibliographystyle{asa}
\bibliography{bibliography.bib}

\newcommand{\noop}[1]{}
\begin{thebibliography}{59}
\newcommand{\enquote}[1]{``#1''}
\expandafter\ifx\csname natexlab\endcsname\relax\def\natexlab#1{#1}\fi

\bibitem[{Audibert and Tsybakov(2007)}]{audibert2007}
Audibert, J.-Y. and Tsybakov, A.~B. (2007), \enquote{Fast learning rates for
  plug-in classifiers,} \textit{The Annals of {S}tatistics}, 35, 608--633.

\bibitem[{Ben-Michael et~al.(2024)Ben-Michael, Imai, and
  Jiang}]{benmichael_asymm_2024}
Ben-Michael, E., Imai, K., and Jiang, Z. (2024), \enquote{Policy {Learning}
  with {Asymmetric} {Counterfactual} {Utilities},} \textit{Journal of the
  American Statistical Association}, 0, 1--14.

\bibitem[{Bickel et~al.(1993)Bickel, Klaassen, Ritov, and Wellner}]{bkrw1993}
Bickel, P., Klaassen, C., Ritov, Y., and Wellner, J. (1993), \textit{Efficient
  and adaptive estimation for semiparametric models}, Johns Hopkins University
  Press Baltimore.

\bibitem[{Bickel and Ritov(1988)}]{bickel1988}
Bickel, P.~J. and Ritov, Y. (1988), \enquote{Estimating integrated squared
  density derivatives: sharp best order of convergence estimates,}
  \textit{Sankhy{\=a}: The Indian Journal of Statistics, Series A}, 381--393.

\bibitem[{Bonvini and Kennedy(2022)}]{bonvini2022}
Bonvini, M. and Kennedy, E.~H. (2022), \enquote{Sensitivity analysis via the
  proportion of unmeasured confounding,} \textit{Journal of the American
  Statistical Association}, 117, 1540--1550.

\bibitem[{Brand and Xie(2010)}]{brand2010}
Brand, J.~E. and Xie, Y. (2010), \enquote{Who benefits most from college?
  Evidence for negative selection in heterogeneous economic returns to higher
  education,} \textit{American sociological review}, 75, 273--302.

\bibitem[{Carneiro et~al.(2003)Carneiro, Hansen, and Heckman}]{carneiro2003}
Carneiro, P., Hansen, K.~T., and Heckman, J.~J. (2003), \enquote{Estimating
  Distributions of Treatment Effects with an Application to the Returns to
  Schooling and Measurement of the Effects of Uncertainty on College,} .

\bibitem[{Carneiro et~al.(2001)Carneiro, Heckman, and Vytlacil}]{carneiro2001}
Carneiro, P., Heckman, J., and Vytlacil, E. (2001), \enquote{Estimating the
  return to education when it varies among individuals,} Tech. rep., mimeo.

\bibitem[{Carneiro et~al.(2011)Carneiro, Heckman, and Vytlacil}]{carneiro2011}
Carneiro, P., Heckman, J.~J., and Vytlacil, E.~J. (2011), \enquote{Estimating
  marginal returns to education,} \textit{American Economic Review}, 101,
  2754--2781.

\bibitem[{Chakraborty and Moodie(2013)}]{chakraborty2013}
Chakraborty, B. and Moodie, E.~E. (2013), \enquote{Statistical methods for
  dynamic treatment regimes,} \textit{Springer-Verlag}, 10, 978--1.

\bibitem[{Chernozhukov et~al.(2018)Chernozhukov, Chetverikov, Demirer, Duflo,
  Hansen, Newey, and Robins}]{chernozhukov2018}
Chernozhukov, V., Chetverikov, D., Demirer, M., Duflo, E., Hansen, C., Newey,
  W., and Robins, J. (2018), \enquote{Double/debiased machine learning for
  treatment and structural parameters,} .

\bibitem[{Chernozhukov et~al.(2009)Chernozhukov, Fernandez-Val, and
  Galichon}]{chernozhukov2009}
Chernozhukov, V., Fernandez-Val, I., and Galichon, A. (2009),
  \enquote{Improving point and interval estimators of monotone functions by
  rearrangement,} \textit{Biometrika}, 96, 559--575.

\bibitem[{Cole and Stuart(2010)}]{cole2010}
Cole, S.~R. and Stuart, E.~A. (2010), \enquote{Generalizing evidence from
  randomized clinical trials to target populations: the ACTG 320 trial,}
  \textit{American journal of epidemiology}, 172, 107--115.

\bibitem[{D'Adamo(2021)}]{DAdamo_2023}
D'Adamo, R. (2021), \enquote{Orthogonal Policy Learning Under Ambiguity,}
  \textit{arXiv preprint arXiv:2111.10904}.

\bibitem[{Dahabreh et~al.(2020)Dahabreh, Petito, Robertson, Hern{\'a}n, and
  Steingrimsson}]{dahabreh2020}
Dahabreh, I.~J., Petito, L.~C., Robertson, S.~E., Hern{\'a}n, M.~A., and
  Steingrimsson, J.~A. (2020), \enquote{Towards causally interpretable
  meta-analysis: transporting inferences from multiple randomized trials to a
  new target population,} \textit{Epidemiology (Cambridge, Mass.)}, 31, 334.

\bibitem[{Danbury et~al.(2015)Danbury, Gould, Baudouin, Berry, Bolton,
  Borthwick, et~al.}]{danbury2015}
Danbury, C., Gould, T., Baudouin, S., Berry, A., Bolton, S., Borthwick, M.,
  et~al. (2015), \textit{Guidelines for the Provision of Intensive Care
  Services}, London: The Intensive Care Society.

\bibitem[{Foster and Syrgkanis(2023)}]{foster2023}
Foster, D.~J. and Syrgkanis, V. (2023), \enquote{Orthogonal statistical
  learning,} \textit{The Annals of Statistics}, 51, 879--908.

\bibitem[{Gy{\"o}rfi et~al.(2002)Gy{\"o}rfi, Kohler, Krzyzak, Walk,
  et~al.}]{gyorfi2002}
Gy{\"o}rfi, L., Kohler, M., Krzyzak, A., Walk, H., et~al. (2002), \textit{A
  distribution-free theory of nonparametric regression}, vol.~1, Springer.

\bibitem[{Haneuse and Rotnitzky(2013)}]{haneuse2013}
Haneuse, S. and Rotnitzky, A. (2013), \enquote{Estimation of the effect of
  interventions that modify the received treatment,} \textit{Statistics in
  medicine}, 32, 5260--5277.

\bibitem[{Harris et~al.(2018)Harris, Singer, Sanderson, Grieve, Harrison, and
  Rowan}]{harris2018}
Harris, S., Singer, M., Sanderson, C., Grieve, R., Harrison, D., and Rowan, K.
  (2018), \enquote{Impact on mortality of prompt admission to critical care for
  deteriorating ward patients: an instrumental variable analysis using critical
  care bed strain,} \textit{Intensive care medicine}, 44, 606--615.

\bibitem[{Harrison et~al.(2007)Harrison, Parry, Carpenter, Short, and
  Rowan}]{harrison2007}
Harrison, D.~A., Parry, G.~J., Carpenter, J.~R., Short, A., and Rowan, K.
  (2007), \enquote{A new risk prediction model for critical care: the Intensive
  Care National Audit \& Research Centre (ICNARC) model,} \textit{Critical care
  medicine}, 35, 1091--1098.

\bibitem[{Imai and Li(2023)}]{imai2023}
Imai, K. and Li, M.~L. (2023), \enquote{Experimental evaluation of
  individualized treatment rules,} \textit{Journal of the American Statistical
  Association}, 118, 242--256.

\bibitem[{Kallus(2022)}]{kallus_harm_2022}
Kallus, N. (2022), \enquote{What’s the {Harm}? {Sharp} {Bounds} on the
  {Fraction} {Negatively} {Affected} by {Treatment},} in \textit{36th
  {Conference} on {Neural} {Information} {Processing} {Systems}}.

\bibitem[{Keele et~al.(2019)Keele, Harris, and Grieve}]{keele2019}
Keele, L., Harris, S., and Grieve, R. (2019), \enquote{Does transfer to
  intensive care units reduce mortality? A comparison of an instrumental
  variables design to risk adjustment,} \textit{Medical care}, 57, e73--e79.

\bibitem[{Kennedy(2022)}]{kennedy2022review}
Kennedy, E.~H. (2022), \enquote{Semiparametric doubly robust targeted double
  machine learning: a review,} \textit{arXiv preprint arXiv:2203.06469}.

\bibitem[{Kennedy(2023)}]{kennedy2023}
--- (2023), \enquote{Towards optimal doubly robust estimation of heterogeneous
  causal effects,} \textit{Electronic Journal of Statistics}, 17, 3008--3049.

\bibitem[{Kennedy et~al.(2020)Kennedy, Balakrishnan, and
  G’Sell}]{kennedy2020b}
Kennedy, E.~H., Balakrishnan, S., and G’Sell, M. (2020), \enquote{Sharp
  instruments for classifying compliers and generalizing causal effects,}
  \textit{The Annals of Statistics}, 48, 2008--2030.

\bibitem[{Kosorok(2008)}]{kosorok2008}
Kosorok, M.~R. (2008), \textit{Introduction to empirical processes and
  semiparametric inference}, vol.~61, Springer.

\bibitem[{LeDell et~al.(2015)LeDell, Petersen, and van~der Laan}]{ledell2015}
LeDell, E., Petersen, M., and van~der Laan, M. (2015), \enquote{Computationally
  efficient confidence intervals for cross-validated area under the ROC curve
  estimates,} \textit{Electronic journal of statistics}, 9, 1583.

\bibitem[{Levis et~al.(2023)Levis, Bonvini, Zeng, Keele, and
  Kennedy}]{levis2023}
Levis, A.~W., Bonvini, M., Zeng, Z., Keele, L., and Kennedy, E.~H. (2023),
  \enquote{Covariate-assisted bounds on causal effects with instrumental
  variables,} \textit{arXiv preprint arXiv:2301.12106}.

\bibitem[{Luedtke and van~der Laan(2016{\natexlab{a}})}]{luedtke2016b}
Luedtke, A.~R. and van~der Laan, M.~J. (2016{\natexlab{a}}), \enquote{Optimal
  individualized treatments in resource-limited settings,} \textit{The
  {I}nternational {J}ournal of {B}iostatistics}, 12, 283--303.

\bibitem[{Luedtke and van~der Laan(2016{\natexlab{b}})}]{luedtke2016a}
--- (2016{\natexlab{b}}), \enquote{Statistical inference for the mean outcome
  under a possibly non-unique optimal treatment strategy,} \textit{Annals of
  {S}tatistics}, 44, 713.

\bibitem[{McClean et~al.(2024)McClean, Branson, and Kennedy}]{mcclean2024}
McClean, A., Branson, Z., and Kennedy, E.~H. (2024), \enquote{Nonparametric
  estimation of conditional incremental effects,} \textit{Journal of Causal
  Inference}, 12, 20230024.

\bibitem[{Murphy(2003)}]{murphy2003}
Murphy, S.~A. (2003), \enquote{Optimal dynamic treatment regimes,}
  \textit{Journal of the Royal Statistical Society: Series B (Statistical
  Methodology)}, 65, 331--355.

\bibitem[{Qian and Murphy(2011)}]{qian2011}
Qian, M. and Murphy, S.~A. (2011), \enquote{Performance guarantees for
  individualized treatment rules,} \textit{Annals of statistics}, 39, 1180.

\bibitem[{Qiu et~al.(2022)Qiu, Carone, and Luedtke}]{qiu2022}
Qiu, H., Carone, M., and Luedtke, A. (2022), \enquote{Individualized treatment
  rules under stochastic treatment cost constraints,} \textit{Journal of Causal
  Inference}, 10, 480--493.

\bibitem[{Qiu et~al.(2021)Qiu, Carone, Sadikova, Petukhova, Kessler, and
  Luedtke}]{qiu2021}
Qiu, H., Carone, M., Sadikova, E., Petukhova, M., Kessler, R.~C., and Luedtke,
  A. (2021), \enquote{Optimal individualized decision rules using instrumental
  variable methods,} \textit{Journal of the American Statistical Association},
  116, 174--191.

\bibitem[{Radcliffe(2007)}]{radcliffe2007}
Radcliffe, N. (2007), \enquote{Using control groups to target on predicted
  lift: Building and assessing uplift model,} \textit{Direct Marketing
  Analytics Journal}, 14--21.

\bibitem[{Rambachan et~al.(2022)Rambachan, Coston, and Kennedy}]{rambachan2022}
Rambachan, A., Coston, A., and Kennedy, E. (2022), \enquote{Counterfactual risk
  assessments under unmeasured confounding,} \textit{arXiv preprint
  arXiv:2212.09844}.

\bibitem[{Robins et~al.(2008)Robins, Li, Tchetgen, van~der Vaart,
  et~al.}]{robins2008}
Robins, J., Li, L., Tchetgen, E., van~der Vaart, A., et~al. (2008),
  \enquote{Higher order influence functions and minimax estimation of nonlinear
  functionals,} in \textit{Probability and statistics: essays in honor of David
  A. Freedman}, Institute of Mathematical Statistics, vol.~2, pp. 335--422.

\bibitem[{Robins(2004)}]{robins2004}
Robins, J.~M. (2004), \enquote{Optimal structural nested models for optimal
  sequential decisions,} in \textit{Proceedings of the Second Seattle Symposium
  in Biostatistics: analysis of correlated data}, Springer, pp. 189--326.

\bibitem[{{Royal College of Physicians}(2012)}]{london2012}
{Royal College of Physicians} (2012), \textit{National early warning score
  (NEWS): standardising the assessment of acute-illness severity in the NHS},
  RCP, London.

\bibitem[{Stensrud et~al.(2022)Stensrud, Laurendeau, and Sarvet}]{stensrud2022}
Stensrud, M.~J., Laurendeau, J., and Sarvet, A.~L. (2022), \enquote{Optimal
  regimes for algorithm-assisted human decision-making,} \textit{arXiv preprint
  arXiv:2203.03020}.

\bibitem[{Sun et~al.(2021)Sun, Munro, Kalashnov, Du, and Wager}]{sun2021}
Sun, H., Munro, E., Kalashnov, G., Du, S., and Wager, S. (2021),
  \enquote{Treatment allocation under uncertain costs,} \textit{arXiv preprint
  arXiv:2103.11066}.

\bibitem[{Takatsu et~al.(2023)Takatsu, Levis, Kennedy, Kelz, and
  Keele}]{takatsu2023}
Takatsu, K., Levis, A.~W., Kennedy, E., Kelz, R., and Keele, L. (2023),
  \enquote{Doubly robust machine learning for an instrumental variable study of
  surgical care for cholecystitis,} \textit{arXiv preprint arXiv:2307.06269}.

\bibitem[{Tsiatis(2007)}]{tsiatis2007}
Tsiatis, A. (2007), \textit{Semiparametric theory and missing data}, Springer
  Science \& Business Media.

\bibitem[{Tsiatis et~al.(2019)Tsiatis, Davidian, Holloway, and
  Laber}]{tsiatis2019}
Tsiatis, A.~A., Davidian, M., Holloway, S.~T., and Laber, E.~B. (2019),
  \textit{Dynamic treatment regimes: Statistical methods for precision
  medicine}, CRC press.

\bibitem[{Tsybakov(2004)}]{tsybakov2004}
Tsybakov, A.~B. (2004), \enquote{Optimal aggregation of classifiers in
  statistical learning,} \textit{The Annals of Statistics}, 32, 135--166.

\bibitem[{Tsybakov(2009)}]{tsybakov2009}
--- (2009), \textit{Introduction to nonparametric estimation}, New York:
  Springer.

\bibitem[{van~der Laan et~al.(2007)van~der Laan, Polley, and
  Hubbard}]{vanderlaan2007}
van~der Laan, M.~J., Polley, E.~C., and Hubbard, A.~E. (2007), \enquote{Super
  {L}earner,} \textit{Statistical applications in genetics and molecular
  biology}, 6.

\bibitem[{van~der Vaart and Wellner(1996)}]{vandervaart1996}
van~der Vaart, A. and Wellner, J.~A. (1996), \textit{Weak Convergence and
  Empirical Processes: With Applications to Statistics}, Springer Verlage.

\bibitem[{van~der Vaart(2000)}]{vandervaart2000}
van~der Vaart, A.~W. (2000), \textit{Asymptotic Statistics}, vol.~3, Cambridge
  University Press.

\bibitem[{Vincent et~al.(1996)Vincent, Moreno, Takala, Willatts,
  De~Mendon{\c{c}}a, Bruining, Reinhart, Suter, and Thijs}]{vincent1996}
Vincent, J.~L., Moreno, R., Takala, J., Willatts, S., De~Mendon{\c{c}}a, A.,
  Bruining, H., Reinhart, C., Suter, P., and Thijs, L.~G. (1996), \enquote{The
  SOFA (Sepsis-related Organ Failure Assessment) score to describe organ
  dysfunction/failure: On behalf of the Working Group on Sepsis-Related
  Problems of the European Society of Intensive Care Medicine (see contributors
  to the project in the appendix),} .

\bibitem[{Willis and Rosen(1979)}]{willis1979}
Willis, R.~J. and Rosen, S. (1979), \enquote{Education and self-selection,}
  \textit{Journal of political Economy}, 87, S7--S36.

\bibitem[{Young et~al.(2014)Young, Hern{\'a}n, and Robins}]{young2014}
Young, J.~G., Hern{\'a}n, M.~A., and Robins, J.~M. (2014),
  \enquote{Identification, estimation and approximation of risk under
  interventions that depend on the natural value of treatment using
  observational data,} \textit{Epidemiologic {M}ethods}, 3, 1--19.

\bibitem[{Zeng et~al.(2023)Zeng, Kennedy, Bodnar, and Naimi}]{zeng2023}
Zeng, Z., Kennedy, E.~H., Bodnar, L.~M., and Naimi, A.~I. (2023),
  \enquote{Efficient generalization and transportation,} \textit{arXiv preprint
  arXiv:2302.00092}.

\bibitem[{Zhao et~al.(2012)Zhao, Zeng, Rush, and Kosorok}]{zhao2012}
Zhao, Y., Zeng, D., Rush, A.~J., and Kosorok, M.~R. (2012), \enquote{Estimating
  individualized treatment rules using outcome weighted learning,}
  \textit{Journal of the American Statistical Association}, 107, 1106--1118.

\bibitem[{Zheng and van~der Laan(2010)}]{zheng2010}
Zheng, W. and van~der Laan, M.~J. (2010), \enquote{Asymptotic theory for
  cross-validated targeted maximum likelihood estimation,} \textit{U.C.
  Berkeley Division of Biostatistics Working Paper Series}.

\bibitem[{Zhou and Xie(2020)}]{zhou2020}
Zhou, X. and Xie, Y. (2020), \enquote{Heterogeneous treatment effects in the
  presence of self-selection: a propensity score perspective,}
  \textit{Sociological Methodology}, 50, 350--385.

\end{thebibliography}

\pagebreak

\begin{appendices}

\section{Proofs of Results in Section~\ref{sec:policy}}
\subsection{Proof of Proposition~\ref{prop:ident}}

Proposition~\ref{prop:ident} is a special case of Proposition~\ref{prop:ident-gen}, stated and proved in Appendix~\ref{app:sens}.

\subsection{Proof of Proposition~\ref{prop:optimal}}
  The first thing to show is that for any fixed $\Delta: \mathcal{X} \to [0,1]$, the optimal unconstrained rule $h^* = \mathds{1}(\tau > 0)$ optimizes $\mathbb{E}[Y(d(\Delta, h))]$ over all possible policies $h$. To see this, observe by Proposition~\ref{prop:ident} that for arbitrary policy $h: \mathcal{X} \to \{0,1\}$,
  \[\mathbb{E}[Y(d(\Delta, h^*))] - \mathbb{E}[Y(d(\Delta, h))] = \mathbb{E}(\Delta \tau (h^* - h)).\]
  But note that by definition of $h^*$,
  \[\tau (h^* - h) = h^* \tau(1 - h) - (1 - h^*)\tau h= |\tau| \{h^*(1 - h) + (1 - h^*)h\} \geq 0, \]
  from which it immediately follows that $\mathbb{E}[Y(d(\Delta, h^*))] \geq \mathbb{E}[Y(d(\Delta, h))]$.

  Now, the remaining goal is to solve a constrained linear optimization problem over a large function class. Specifically, by Proposition~\ref{prop:ident} together with what we just showed, the objective to maximize is $\mathbb{E}(\Delta\tau(h^* - \pi)) = \mathbb{E}(\Delta\beta)$ over $\left\{\Delta : \mathcal{X} \to [0,1] \mid \mathbb{E}(\Delta(X)) \leq \delta\right\}$.
  One way to prove the stated result is to characterize the solution over a finite set of decision variables when $X$ is discrete, then generalize to arbitrary $X$ via a measure-theoretic limiting argument. We will proceed instead, with the benefit of hindsight, by directly comparing the value of the objective function for an arbitrary contact rule $\Delta$ to that of $\Delta_{\delta}^*$. The argument closely follows that of Theorem 1 in \citet{kennedy2020b}.

  \vspace{3mm}
  Let $\Delta: \mathcal{X} \to [0,1]$ be an arbitrary contact
  rule such that $\mathbb{E}(\Delta(X)) \leq \delta$. Observe that,
  \begin{equation} \label{eq:diff}
    \Delta_{\delta}^*(X) - H_{\Delta} = \mathds{1}(\beta(X) > q_{1 - \delta})
    \left\{1 - H_{\Delta}\right\} - \mathds{1}(\beta(X) \leq q_{1 -
      \delta})H_{\Delta},
  \end{equation}
  hence, taking expectations,
  \begin{equation} \label{eq:diffex}
    0 \leq \delta \mathbb{E}(1 - H_{\Delta} \mid \beta(X) > q_{1 -
      \delta}) - (1 - \delta) \mathbb{E}(H_{\Delta} \mid \beta(X)
    \leq q_{1 - \delta}) ,
  \end{equation}
  since
  $\mathbb{E}(\Delta_{\delta}^*(X)) = \delta \geq \mathbb{E}(\Delta(X))$ by
  construction. Therefore, using \eqref{eq:diff},
  \begin{align*}
    &\mathbb{E}(\Delta_{\delta}^*(X) \beta(X)) - \mathbb{E}(\Delta(X) \beta(X)) \\
    &= \mathbb{E}(\beta(X) \{\Delta_{\delta}^*(X) - H_{\Delta}\}) \\
    &= \mathbb{E}(\beta(X)\{\mathds{1}(\beta(X) > q_{1 - \delta})
      \left\{1 - H_{\Delta}\right\} - \mathds{1}(\beta(X) \leq q_{1 -
      \delta})H_{\Delta}\}) \\
    & \geq q_{1 - \delta}\left\{\delta \mathbb{E}(1 - H_{\Delta}
      \mid \beta(X) > q_{1 - \delta}) - (1 - \delta)
      \mathbb{E}(H_{\Delta} \mid \beta(X)
      \leq q_{1 - \delta})\right\} \\
    & \geq 0,
  \end{align*}
  using \eqref{eq:diffex}. This proves the result.

\subsection{Proof of Corollary~\ref{cor:gap}}
The policy $d(1, h^*)$ is the same as $h^*$, and is optimal in the context of Proposition~\ref{prop:optimal} under the vacuous budget constraint that $\mathbb{E}(\Delta(X)) \leq 1$. So by this last result,
\[\mathbb{E}(Y(h^*)) - \mathbb{E}[Y(d(\Delta_{\delta}^*, h^*))] = \mathbb{E}(\beta + Y - \Delta_{\delta}^*\beta - Y) = \mathbb{E}((1 - \Delta_{\delta}^*)\beta)).\]
The bound is obtained by noting that $1 - \Delta_{\delta}^* = \mathds{1}(\beta \leq q_{1 - \delta})$, so
\[\mathbb{E}(Y(h^*)) - \mathbb{E}[Y(d(\Delta_{\delta}^*, h^*))] = \mathbb{E}\left(\mathds{1}(\beta \leq q_{1 - \delta}) \beta\right) \leq q_{1 - \delta}\mathbb{P}[\beta \leq q_{1 - \delta}] = q_{1 - \delta}(1 - \delta),\]
when the quantile is unique.

\subsection{Proof of Proposition~\ref{prop:AUPBC}}
Observe that
\[
    \mathcal{A}=\int_0^1 \mathbb{E}\left((\Delta_{\delta}^* - \delta)\beta\right) \, d \delta
    = \mathbb{E}\left(\beta\int_0^1 \{\Delta_{\delta}^* - \delta\} \, d \delta\right)
\]
by Fubini's theorem (assuming implicitly that $\mathbb{E}(\beta) < \infty$). But note that since $F_{\beta}$ is strictly monotone, $q_{1 - \delta} = F^{-1}_{\beta}(1 - \delta)$ and $\beta > q_{1 - \delta} \iff F_{\beta}(\beta) > 1 - \delta$. Thus,
\[\int_0^1 \{\Delta_{\delta}^* - \delta\} \, d \delta = \int_0^1 \mathds{1}(\beta > q_{1 - \delta})\, d \delta - \frac{1}{2} = \int_{1 - F_{\beta}(\beta)}^1\, d \delta - \frac{1}{2} = F_{\beta}(\beta) - \frac{1}{2}.\]
This shows that $\mathcal{A} = \mathbb{E}\left(\beta\left\{F_{\beta}(\beta) - \frac{1}{2}\right\}\right) = \mathrm{Cov}(\beta, F_{\beta}(\beta))$, since our assumptions on $F_{\beta}$ imply $F_\beta(\beta) \sim \mathrm{Unif}(0,1)$ and $\mathbb{E}(F_{\beta}(\beta)) = \frac{1}{2}$.

\section{Sensitivity to Unmeasured Confounders}\label{app:sens}
In this appendix, we consider potential violations of the no unmeasured confounding assumption (Assumption~\ref{ass:NUC}), and derive sharp bounds on the value $\mathbb{E}(Y(d(\Delta_{\delta}^*, h^*)))$ under an outcome-based sensitivity model. We also provide a characterization of the true optimal treatment rule $d(\Delta_{\delta}^{\dagger}, h^{\dagger})$ in the absence of Assumption~\ref{ass:NUC}, and bound the gap between its value and that of $d(\Delta_{\delta}^*, h^*)$ in the same sensitivity model.

Throughout this appendix, we will work under Assumptions~\ref{ass:consistency} and \ref{ass:positivity}---positivity is required for $\mu_a(x)$ to be well-defined for each $x$. Let $\nu_a(X) = \mathbb{E}(Y(a) \mid X, A = 1 - a)$, for $a \in \{0,1\}$, so that $\mu_1^{\dagger} - \mu_1 = (1 - \pi)(\nu_1 - \mu_1)$ and $\mu_0^{\dagger} - \mu_0 = \pi(\nu_0 - \mu_0)$. We begin by providing a general result which implies Proposition~\ref{prop:ident}, but is also useful for the ensuing sensitivity analysis results in this appendix.

\begin{proposition}\label{prop:ident-gen}
    Under Assumption~\ref{ass:consistency}, and given a contact rule $\Delta: \mathcal{X} \to [0,1]$ and a policy $h:\mathcal{X}\to\{0,1\}$, the conditional value of the policy $d(\Delta, h)$, defined by~\eqref{eq:contact-policy}, is given by
    \[\mathbb{E}[Y(d(\Delta, h)) \mid X]= \Delta \left\{h \tau^{\dagger} + \mu_0^{\dagger} - m\right\} + m,\]
    where $m(X) = \mathbb{E}(Y \mid X)$. If we further assert Assumptions~\ref{ass:positivity} and \ref{ass:NUC}, then the conditional value is identified by $\Delta\tau (h - \pi) + m$, so that the marginal value of the policy $d(\Delta, h)$ is 
    \[\mathbb{E}[Y(d(\Delta, h))] = \mathbb{E}\left(\Delta\tau (h - \pi) + Y\right).\]
\end{proposition}

\begin{proof}
    Note that by consistency (i.e., Assumption~\ref{ass:consistency}),
  \begin{align*}
      Y(d(\Delta, h)) &= Y(1) \left\{H_{\Delta} h(X) + \left[1 -
        H_{\Delta}\right]A\right\} + Y(0) \left\{H_{\Delta}\left[1 -
        h(X)\right] + \left[1 - H_{\Delta}\right](1 - A)\right\} \\
        &= H_{\Delta}\left\{h(X) Y(1) + (1 - h(X)) Y(0))\right\} + (1 - H_{\Delta})Y,
  \end{align*}
  so that, as $H_{\Delta} \ind (A, Y(0), Y(1)) \mid X$,
  \begin{align*}
    & \mathbb{E}(Y(d(\Delta, h))) \mid X) \\
    &= \Delta(X)\left\{h(X) \mathbb{E}(Y(1) \mid X) + (1 - h(X)) \mathbb{E}(Y(0) \mid X)\right\} + (1-\Delta(X))\mathbb{E}(Y \mid X) \\
    &= \Delta(X)\left\{h(X) \mu_1^{\dagger}(X) + (1 - h(X))\mu_{0}^{\dagger}(X)\right\} + (1 - \Delta(X))m(X) \\
    &= \Delta(X) \left\{h(X) \tau^{\dagger}(X) + \mu_0^{\dagger}(X) - m(X)\right\} + m(X),
  \end{align*}
  where the second equality is by definition of $\mu_a^{\dagger}$ and $m$, and the third by rearranging. When positivity (i.e., Assumption~\ref{ass:positivity}) holds, then $m = \pi\mu_1 + (1 - \pi)\mu_0 = \pi\tau + \mu_0$, so we further have, omitting inputs,
  \[
    \mathbb{E}(Y(d(\Delta, h))) \mid X)
    = \Delta \left\{h\big(\tau^{\dagger} - \tau\big) + \mu_0^{\dagger} - \mu_0\right\} + \Delta \tau \left\{h - \pi\right\} + m.\]
    Lastly, when Assumption~\ref{ass:NUC} holds, we have $\tau^{\dagger} = \tau$ and $\mu_0^{\dagger} = \mu_0$, so
    \[ \mathbb{E}(Y(d(\Delta, h))) \mid X) = \Delta \tau \left\{h - \pi\right\} + m \implies \mathbb{E}[Y(d(\Delta, h))]  = \mathbb{E}\left(\Delta \tau \left\{h - \pi\right\} + Y\right),\]
  by iterated expectations.
\end{proof}

For a sensitivity parameter $\Gamma \geq 0$, we consider the following outcome-based sensitivity model
\begin{equation}\label{eq:sens-outcome}
|\nu_a - \mu_a| \leq \Gamma, \text{ for } a \in \{0,1\}.
\end{equation}
That is, for a fixed $\Gamma$, model \eqref{eq:sens-outcome} says that conditional mean outcomes for the counterfactual $Y(a)$ differ by at most $\Gamma$ between treatment groups. The following result gives sharp bounds for $\mathbb{E}(Y(d(\Delta_{\delta}^*, h^*)))$ under this model.

\begin{proposition}\label{prop:sens-bounds}
    Under Assumptions~\ref{ass:consistency} and \ref{ass:positivity}, and sensitivity model~\eqref{eq:sens-outcome},
    \[\mathbb{E}(Y(d(\Delta_{\delta}^*, h^*))) \in \big[V_{\delta} -\Gamma\mathbb{E}(\Delta_{\delta}^* c), V_{\delta} +\Gamma\mathbb{E}(\Delta_{\delta}^* c)\big],\]
    where $c = h^*(1 - \pi) + (1 - h^*)\pi = \mathbb{P}[A \neq h^*(X) \mid X]$ as defined in Section~\ref{sec:estimands}. Without further conditions, these bounds are sharp.
\end{proposition}

\begin{proof}
    Note that by the first part of Proposition~\ref{prop:ident-gen},
    \begin{align*}
        \left|\mathbb{E}(Y(d(\Delta_{\delta}^*, h^*))) - V_{\delta}\right|  &= \left|\mathbb{E}\left(\Delta_{\delta}^*\left\{h^*(\mu_1^{\dagger} - \mu_1) + (1 - h^*)(\mu_0^{\dagger} - \mu_0)\right\}\right)\right| \\
        &= \left|\mathbb{E}\left(\Delta_{\delta}^*\left\{h^*(1 - \pi)(\nu_1 - \mu_1) + (1 - h^*)\pi(\nu_0 - \mu_0)\right\}\right)\right| \\
        & \leq \Gamma\mathbb{E}(\Delta_{\delta}^* c),
    \end{align*}
    by model~\eqref{eq:sens-outcome} and the definition of $c$. To see that this bound is attainable, we can trivially construct a data generating mechanism for $(X, A, Y(0), Y(1))$ compatible with the observed data distribution $\mathbb{P}$, such that either $\nu_a - \mu_a \equiv - \Gamma$ or $\nu_a - \mu_a \equiv \Gamma$ for both $a$.
\end{proof}

The bounds in Proposition~\ref{prop:sens-bounds} take a relatively simple form, and estimation can proceed almost exactly as described for $\widehat{V}_{\delta}$ in Section~\ref{sec:estimation-value}, but with extra terms. To obtain efficient similar asymptotic behavior as described in Theorem~\ref{thm:est-value} for $\widehat{V}_{\delta}$, one should take
\[\widehat{V}_{\delta} \pm \Gamma\mathbb{P}_n \left[\widehat{\Delta}_{\delta}^* \left\{\widehat{h}^* + (1 - 2\widehat{h}^*)(A - \widehat{\pi})\right\}\right].\]
This results in similar remainder terms and convergence rates, though we omit details here.

One feature of these bounds is that they collapse to the identified functional $V_{\delta}$ when either $\Gamma \to 0$ or $\delta \to 0$---note that $\mathbb{E}(\Delta_{\delta}^* c) \leq \mathbb{E}(\Delta_{\delta}^*) = \delta$. The former property is obvious, but the latter is especially nice in settings with low budgets. In particular, when $\delta$ is small, the performance of the treatment rule $d(\Delta_{\delta}^*, h^*)$ will not differ much from what we would estimate under Assumption~\ref{ass:NUC}, even under fairly substantial violations.

While inference about the putative treatment rule $d(\Delta_{\delta}^*, h^*)$ may be of interest, one may also wish to know the actual optimal rule when Assumption~\ref{ass:NUC} is violated. Our next result gives this characterization.

\begin{proposition}\label{prop:sens-optimal}
    Under Assumption~\ref{ass:consistency}, define $h^{\dagger} = \mathds{1}(\tau^{\dagger} > 0)$ and $\beta^{\dagger} = h^{\dagger}\tau^{\dagger} + \mu_{0}^{\dagger} - m$, and assume the $(1-\delta)$-quantile of $\beta^{\dagger}$, $q_{1-\delta}^{\dagger}$, is unique. Then the contact rule $\Delta_{\delta}^{\dagger} = \mathds{1}(\beta^{\dagger} > q_{1 - \delta}^{\dagger})$ and policy $h^{\dagger}$ jointly maximize the value $\mathbb{E}[Y(d(\Delta, h))]$ over all rules $d(\Delta, h)$ of the form~\eqref{eq:contact-policy}, subject to the constraint that $\mathbb{E}(\Delta(X)) \leq \delta$, with maximal value $\mathbb{E}(\Delta_{\delta}^{\dagger} \beta^{\dagger} + Y)$.
\end{proposition}
\begin{proof}
    This follows by the first part of Proposition~\ref{prop:ident-gen}, and the same argument as the proof of Proposition~\ref{prop:optimal}.
\end{proof}

Proposition~\ref{prop:sens-optimal} tells us two things about the setting where there are unmeasured confounders: (1) the optimal unconstrained policy based on covariates $X$ would be a function (namely $h^{\dagger}$) of true CATE, $\tau^{\dagger}$, and (2) the optimal contact rule would be based on the true CPB, $\beta^{\dagger} = \mathbb{E}(Y(h^{\dagger}) - Y \mid X)$. Both of these are, of course, not identified when Assumption~\ref{ass:NUC} is violated. In practice, we may wish to know how much the observational treatment rule $d(\Delta_{\delta}^*, h^*)$ differs from the optimal rule $d(\Delta_{\delta}^\dagger, h^\dagger)$ in terms of value. Our final result in this section provides a bound on this gap.

\begin{proposition}\label{prop:sens-optimal-gap}
    Under Assumptions~\ref{ass:consistency} and \ref{ass:positivity}, and model~\eqref{eq:sens-outcome},
    \[\left|\mathbb{E}[Y(d(\Delta_{\delta}^\dagger, h^\dagger))] - \mathbb{E}[Y(d(\Delta_{\delta}^*, h^*))]\right| \leq \Gamma (4 + \delta)\]
\end{proposition}

\begin{proof}
    By the first part of Proposition~\ref{prop:ident-gen},
    \begin{align*}
    &\mathbb{E}[Y(d(\Delta_{\delta}^\dagger, h^\dagger))] - \mathbb{E}[Y(d(\Delta_{\delta}^*, h^*))] \\
    &= \mathbb{E}\left(\Delta_{\delta}^{\dagger} \left\{h^{\dagger} \tau^{\dagger} + \mu_0^{\dagger} - m\right\} - \Delta_{\delta}^* \left\{h^* \tau^{\dagger} + \mu_0^{\dagger} - m\right\}\right) \\
    &= \mathbb{E}\left(\left\{\Delta_{\delta}^{\dagger} - \Delta_{\delta}^*\right\}\beta^{\dagger} + \Delta_{\delta}^* \tau^{\dagger}\left\{h^{\dagger} - h^*\right\}\right) \\
    &= \mathbb{E}\left(\left\{\Delta_{\delta}^{\dagger} - \Delta_{\delta}^*\right\}(\beta^{\dagger} - q_{1 - \delta}^{\dagger})\right) + \mathbb{E}\left(\Delta_{\delta}^* \tau^{\dagger}\left\{h^{\dagger} - h^*\right\}\right),
    \end{align*}
    using that fact that $\mathbb{E}(\Delta_{\delta}^*) = \mathbb{E}(\Delta_{\delta}^\dagger) = \delta$.
    By Lemma~\ref{lemma:diff}, and \eqref{eq:sens-outcome},
    $|\tau^{\dagger}||h^{\dagger} - h^*| \leq |\tau^{\dagger} - \tau| \leq \Gamma$, so $\mathbb{E}\left(\Delta_{\delta}^* \tau^{\dagger}\left\{h^{\dagger} - h^*\right\}\right) \leq \Gamma \delta$. Similarly, by Lemma~\ref{lemma:diff},
    \[\left|\Delta_{\delta}^{\dagger} - \Delta_{\delta}^*\right|\left|\beta^{\dagger} - q_{1 - \delta}^{\dagger}\right| \leq \left|\beta^{\dagger} - \beta\right| + \left|q_{1 - \delta}^{\dagger} - q_{1 - \delta}\right| \leq 4 \Gamma,\]
    since $\lVert \beta^{\dagger} - \beta \rVert_{\infty} \leq 2\Gamma$ under~\eqref{eq:sens-outcome}, which also implies that $|q_{1 - \delta}^{\dagger} - q_{1 - \delta}| \leq 2 \Gamma$.
\end{proof}

The gap described in Proposition~\ref{prop:sens-optimal-gap} may be quite wide, and we do not claim that these bounds are sharp---we suspect these can be improved with more careful arguments. 

We lastly note that in this setting with unmeasured confounders, one may argue that contact rules and policies based solely on covariates $X$ do not make use of all available information. Indeed, when Assumption~\ref{ass:NUC} is violated, one can consider ``superoptimal'' regimes that incorporate the natural value of treatment itself, $A$ \citep{stensrud2022}. Intuitively, the natural value of treatment is a proxy for unmeasured confounders, and can therefore be used for more targeted contact/treatment selection. In some settings it may be infeasible to know what $A$ is yet also intervene before a subject is treated with $A$, while in other settings it may be plausible, e.g., when a doctor knows which treatment they would have prescribed in the absence of any external intervention (see \citet{haneuse2013,young2014} for more discussion). In any case, we leave further exploration of this superoptimality issue in the context of our proposed treatment rules for future research.

\section{Proofs of Results in Section~\ref{sec:estimation}}

\subsection{Proof of Lemma~\ref{lemma:pseudo-bias}}
Observe that
\begin{align*}
    &\mathbb{E}(\phi(O; \widetilde{{\mathbb{P}}}) - \phi(O;{\mathbb{P}}) \mid X) \\
    &= \mathbb{E}\left(\left(\widetilde{h}^* - \widetilde{\pi}\right)\left\{\frac{A}{\widetilde{\pi}} - \frac{1 - A}{1 - \widetilde{\pi}}\right\}\left(Y - \widetilde{\mu}_A\right) + \widetilde{\tau}\left(\widetilde{h}^* - A\right) \, \bigg| \, X\right) - \beta\\
    &= \left(\widetilde{h}^* - \widetilde{\pi}\right)\left\{\frac{\pi}{\widetilde{\pi}}\left(\mu_1 - \widetilde{\mu}_1\right) - \frac{1 - \pi}{1 - \widetilde{\pi}}\left(\mu_0 - \widetilde{\mu}_0\right)\right\} + \widetilde{\tau}\left(\widetilde{h}^* - \pi\right) - \tau(h^* - \pi) \\
    & = \left(\widetilde{h}^* - \widetilde{\pi}\right)\left\{
      \frac{\pi}{\widetilde{\pi}}\left(\mu_1 - \widetilde{\mu}_1\right) -
      \frac{1 - \pi}{1 - \widetilde{\pi}}
      \left(\mu_0 -\widetilde{\mu}_0\right) + \widetilde{\tau} - \tau\right\} \\
      & \quad \quad +
      \widetilde{\tau}\left(\widetilde{\pi} - \pi\right) +
      \tau\left(\{\widetilde{h}^* - h^*\} - \{\widetilde{\pi} - \pi\}\right) \\
    &= \left(\widetilde{h}^* - \widetilde{\pi}\right) \sum_{j=0}^1
      \frac{\{\widetilde{\pi} - \pi\}
      \{\widetilde{\mu}_j - \mu_j\}}{j\widetilde{\pi} + (1 - j)(1 - \widetilde{\pi})} + \{\widetilde{\tau} - \tau\}\{\widetilde{\pi} - \pi\}
      + \{\widetilde{h}^* - h^*\}\tau.
\end{align*}
The first equality is by definition of $\phi$, the second by conditioning on $(A, X)$ then $X$, and the third and fourth by rearranging.

\subsection{Proof of Theorem~\ref{thm:est-value}}
 We may decompose the error as follows:
  \begin{align*}
    \widehat{V}_{\delta} - V_{\delta}
    &= \left(\mathbb{P}_n -
      {\mathbb{P}}\right)\left(\Delta_{\delta}^*(X)\phi(O;{\mathbb{P}}) + Y\right) + \underbrace{\left(\mathbb{P}_n -
      {\mathbb{P}}\right)\left(\widehat{\Delta}_{\delta}^*(X)\phi(O;\widehat{\mathbb{P}}) -
      \Delta_{\delta}^*(X)\phi(O;{\mathbb{P}})\right)}_{\mathrm{(i)}} \\
    & \quad \quad
      + \underbrace{\mathbb{P}\left(\widehat{\Delta}_{\delta}^*(X)\phi(O;\widehat{{\mathbb{P}}}) -
      \Delta_{\delta}^*(X)\phi(O;{\mathbb{P}})\right)}_{\mathrm{(ii)}}
  \end{align*}
  By showing that the empirical process term (i) is $o_{\mathbb{P}}(n^{-1/2})$, and that the bias term (ii) is $-q_{1-\delta}\left(\mathbb{P}_n - {\mathbb{P}}\right)\Delta_{\delta}^*(X) + O_{\mathbb{P}}(R_{1,n} + R_{2,n} + R_{3,n}) + o_{\mathbb{P}}(n^{-1/2})$, the desired result will follow.

  Starting with term (i), observe that
  \begin{align*}
      \left(\mathbb{P}_n - {\mathbb{P}}\right)\left(\widehat{\Delta}_{\delta}^*(X)\phi(O;\widehat{\mathbb{P}}) -\Delta_{\delta}^*(X)\phi(O;{\mathbb{P}})\right)
  &=  \left(\mathbb{P}_n - {\mathbb{P}}\right)\left(\left\{\widehat{\Delta}_{\delta}^*(X)-\Delta_{\delta}^*(X)\right\}\phi(O;{\widehat{\mathbb{P}}})\right)\\
  & \quad \quad + \left(\mathbb{P}_n - {\mathbb{P}}\right)\left(\Delta_{\delta}^*(X)\left\{\phi(O;\widehat{\mathbb{P}}) - \phi(O; \mathbb{P})\right\}\right)
  \end{align*}
  The first of these terms is shown to be $o_{\mathbb{P}}(n^{-1/2})$ in Lemma~\ref{lemma:quantile-emp-process}. For the second,
  observe that for any $t > 0$,
\begin{align*}
  \lVert \widehat{h}^* - h^*\rVert^2
  &= \mathbb{P}(\{\widehat{h}^* - h^*\}^2) \\
  &= \mathbb{P}(|\mathds{1}(\widehat{\tau} > 0) - \mathds{1}(\tau > 0)|) \\
  & \leq \mathbb{P}\left(\mathds{1}(|\tau| \leq |\widehat{\tau} - \tau|)\right) \\
  &\leq \mathbb{P}\left[|\tau| \leq t\right]
    + \mathbb{P}\left[|\widehat{\tau} - \tau| > t\right] \\
  & \lesssim t^{a} + \frac{1}{t}\lVert \widehat{\tau} - \tau\rVert,
\end{align*}
implying $\lVert \widehat{h}^* - h^*\rVert = o_{\mathbb{P}}(1)$ by Lemma
\ref{lemma:convbound}, as
$\lVert \widehat{\tau} - \tau\rVert \leq \lVert \widehat{\mu}_0 -
\mu_0\rVert + \lVert \widehat{\mu}_1 - \mu_1\rVert = o_{\mathbb{P}}(1)$ by
assumption. In the the third line, we used Lemma \ref{lemma:diff}, and
in the last line we used Assumption~\ref{ass:margin},
Markov's inequality, and the fact that
$\lVert \, \cdot \, \rVert_{L_1} \leq \lVert \, \cdot \,
\rVert_{L_2}$. Thus, we have
\begin{align*}
&\left\lVert \Delta_{\delta}^*(X)\left\{\phi(O;\widehat{\mathbb{P}}) -
    \phi(O;\mathbb{P})\right\}\right\rVert \\
    &\leq
  \big\lVert\phi(O;\widehat{\mathbb{P}}) -
    \phi(O;\mathbb{P})\big\rVert \\
  & = \bigg\lVert \left(\widehat{h}^* - \widehat{\pi}\right)
    \left\{\frac{A}{\widehat{\pi}} -
    \frac{1 - A}{1 - \widehat{\pi}}\right\}\left(Y - \widehat{\mu}_A\right) +
    \widehat{\tau}\left(\widehat{h}^* - A\right)  +
    \left(h^* - \pi\right)\left\{\frac{A}{\pi} -
    \frac{1 - A}{1 - \pi}\right\}\left(Y - \mu_A\right) -
    \tau\left(h^* - A\right)\bigg\rVert \\
  & \leq
    \bigg\lVert \left(\widehat{h}^* - \widehat{\pi}\right)
    \left\{\left(\frac{A}{\widehat{\pi}} - \frac{A}{\pi}\right)
    \left(\mu_1 - \widehat{\mu}_1\right)-
    \left(\frac{1 - A}{1 - \widehat{\pi}} - \frac{1 - A}{1 - \pi}\right)
    \left(\mu_0 - \widehat{\mu}_0\right)\right\} +
    \left(\widehat{h}^* - A\right)\left(\widehat{\tau} - \tau\right) \bigg\rVert \\
  & \quad \quad + \bigg\lVert \left(\widehat{h}^* - \widehat{\pi} - h^* +
    \pi\right)\left\{\frac{A}{\pi} -
    \frac{1 - A}{1 - \pi}\right\}\left(Y - \mu_A\right) +
    \left(\widehat{h}^* - h^*\right)
    \left(\widehat{\tau} - \tau\right)\bigg\rVert \\
  & \lesssim \lVert \widehat{\mu}_1 - \mu_1\rVert +
    \lVert \widehat{\mu}_0 - \mu_0\rVert + \lVert \widehat{h}^* - h^*\rVert
    + \lVert \widehat{\pi} - \pi\rVert \\
  &= o_{\mathbb{P}}(1),
  \end{align*}
where we used consistency of $\widehat{\mu}_0$, $\widehat{\mu}_1$, $\widehat{h}^*$, and $\widehat{\pi}$, as well as our boundedness assumptions. Hence, by Lemma~\ref{lemma:emp-process}, the second term above is also $o_{\mathbb{P}}(n^{-1/2})$.

  It remains to analyze the bias term (ii). See that
  \begin{align*}
      &\mathbb{P}\left(\widehat{\Delta}_{\delta}^*(X)\phi(O;\widehat{{\mathbb{P}}}) -
      \Delta_{\delta}^*(X)\phi(O;{\mathbb{P}})\right) \\
      &= \mathbb{P}\left(\widehat{\Delta}_{\delta}^*(X)(\phi(O;\widehat{{\mathbb{P}}}) - \phi(O;\mathbb{P}))\right) + \mathbb{P}\left((\widehat{\Delta}_{\delta}^*(X) - \Delta_{\delta}^*(X))\phi(O;{\mathbb{P}})\right).
  \end{align*}
  By the product bias in Lemma~\ref{lemma:pseudo-bias},
  \begin{align*}
\left|\mathbb{P}\left(\widehat{\Delta}_{\delta}^*(X)(\phi(O;\widehat{{\mathbb{P}}}) - \phi(O;\mathbb{P}))\right)\right| & \leq \mathbb{P}\left[\widehat{\Delta}_{\delta}^*(X) \left|\mathbb{P}\left(\phi(O;\widehat{{\mathbb{P}}}) - \phi(O;\mathbb{P})\mid X\right)\right|\right] \\
      & \lesssim R_{1,n} + \mathbb{P}\left(\big|\widehat{h}^* - h^*\big||\tau|\right),
  \end{align*}
  but by Lemma~\ref{lemma:diff} and Assumption~\ref{ass:margin},
  \begin{align*}
      \mathbb{P}\left(\big|\widehat{h}^* - h^*\big||\tau|\right) \leq \mathbb{P}\left(\mathds{1}(|\tau| \leq |\widehat{\tau} - \tau|)|\widehat{\tau} - \tau|\right) \lesssim \lVert \widehat{\tau} - \tau\rVert_{\infty}^{1 + a} = R_{2,n}.
  \end{align*}
  Meanwhile, by iterated expectations,
  \begin{align*}
      \mathbb{P}\left((\widehat{\Delta}_{\delta}^*(X) - \Delta_{\delta}^*(X))\phi(O;{\mathbb{P}})\right)
      &= \mathbb{P}\left((\widehat{\Delta}_{\delta}^* - \Delta_{\delta}^*)\beta\right) \\
      &= \mathbb{P}\left((\widehat{\Delta}_{\delta}^* - \Delta_{\delta}^*)\{\beta - q_{1 - \delta}\}\right) + q_{1 - \delta}\mathbb{P}\left(\widehat{\Delta}_{\delta}^* - \Delta_{\delta}^*\right),
  \end{align*}
  which can again be partly bounded by Lemma~\ref{lemma:diff} and Assumption~\ref{ass:margin},
  \begin{align*}
      \left|\mathbb{P}\left((\widehat{\Delta}_{\delta}^* - \Delta_{\delta}^*)\{\beta - q_{1 - \delta}\}\right)\right|
      & \leq \mathbb{P}\left(\mathds{1}(|\beta - q_{1 - \delta}| \leq |\widehat{\beta} - \widehat{q}_{1 - \delta} - \beta + q_{1 - \delta}|)|\beta - q_{1 - \delta}|\right)\\
      & \lesssim \left(\lVert \widehat{\beta} - \beta\rVert_{\infty} + |\widehat{q}_{1 - \delta} - q_{1 - \delta}|\right)^{1 + b} \\
      &= R_{3,n}.
  \end{align*}
  Recalling that by construction of $\widehat{q}_{1 - \delta}$, $\mathbb{P}_n\big[\widehat{\Delta}_{\delta}^*\big] = \delta + o_{\mathbb{P}}(n^{-1/2})$, we have
  \begin{align*}
      o_{\mathbb{P}}(n^{-1/2}) &= \mathbb{P}_n\big[\widehat{\Delta}_{\delta}^*\big] - \mathbb{P}\big[\Delta_{\delta}^*\big] \\
      &= (\mathbb{P}_n - \mathbb{P})(\widehat{\Delta}_{\delta}^* - \Delta_{\delta}^*) + (\mathbb{P}_n
    - \mathbb{P})\Delta_{\delta}^* + \mathbb{P}(\widehat{\Delta}_{\delta}^* - \Delta_{\delta}^*) \\
  & = (\mathbb{P}_n
    - \mathbb{P})\Delta_{\delta}^* + \mathbb{P}(\widehat{\Delta}_{\delta}^* - \Delta_{\delta}^*) + o_{\mathbb{P}}(n^{-1/2}),
  \end{align*}
  by Lemma~\ref{lemma:quantile-emp-process}. Hence $q_{1 - \delta}\mathbb{P}\left(\widehat{\Delta}_{\delta}^* - \Delta_{\delta}^*\right) = - q_{1 - \delta}(\mathbb{P}_n - \mathbb{P})\Delta_{\delta}^* + o_{\mathbb{P}}(n^{-1/2})$ which completes the proof.

\subsection{Proof of Theorem~\ref{thm:est-CPB}}
This is obtained immediately from Lemma~\ref{lemma:pseudo-bias} and Lemma~\ref{lemma:pseudo-oracle}.

\subsection{Proof of Proposition~\ref{prop:bias-rule}}
This is an immediate consequence of Proposition~\ref{prop:ident} and the proof of Theorem~\ref{thm:est-value}.

\subsection{Proof of Theorem~\ref{thm:est-AUPBC}}
The result follows almost immediately from Lemma~\ref{lemma:sup-emp-process}, as
\begin{align*}
    &\left|\widehat{\mathcal{A}} - \mathcal{A} - (\mathbb{P}_n - \mathbb{P})N\right| \\
    &= \bigg|\int_0^1 \{\widehat{V}_{\delta} - V_{\delta} - (\mathbb{P}_n - \mathbb{P})(\Delta_\delta^*\{\phi(O;\mathbb{P}) - q_{1 -\delta}\} + Y)\} \, d\delta \\
    & \quad \quad - \frac{1}{2}\{\widehat{V}_1 - V_1 - (\mathbb{P}_n - \mathbb{P})(\phi(O;\mathbb{P})+ Y)\}\bigg| \\
    &\leq \sup_{\delta \in [0,1]} \left|\widehat{V}_{\delta} - V_{\delta} - (\mathbb{P}_n - \mathbb{P})(\Delta_{\delta}^*\{\phi(O; \mathbb{P}) - q_{1 - \delta}\} + Y)\right| \\
    & \quad \quad + \frac{1}{2}\left|\widehat{V}_1 - V_1 - (\mathbb{P}_n - \mathbb{P})(\phi(O;\mathbb{P})+ Y)\right|
  \end{align*}
  The first term is $O_{\mathbb{P}}(R_{1,n} + R_{2,n} + R_{3,n}^*) + o_{\mathbb{P}}(n^{-1/2})$ by Lemma~\ref{lemma:sup-emp-process}, while the second term is $O_{\mathbb{P}}(R_{1,n} + R_{2,n}) + o_{\mathbb{P}}(n^{-1/2})$ by the proof of Theorem~\ref{thm:est-value}. The result for the normalized AUPBC estimator $\widehat{\overline{\mathcal{A}}}$ follows similarly, working with quotients exactly as in Lemmas 3.1, 3.2, and S.1 of \citet{takatsu2023}.
 

\section{General Pseudo-Outcome Results under $L_2(\mathbb{P})$ Stability}\label{app:L2}

In this appendix, we review the oracle inequality results for two-stage pseudo-outcome regression in \cite{kennedy2023}, extended to deal with $L_2(\mathbb{P})$ error in \citet{rambachan2022}.

The general setting is that we observe data $O \sim \mathbb{P}$, which includes covariates $X$, and we aim to estimate a function $m(X) = \mathbb{E}(f(O) \mid X)$. Leveraging sample splitting, we obtain an estimate $\widehat{f}$ using training data, then regress $\widehat{f}$ on covariates $X$ in test data. We will see that under mild conditions on the second stage regression estimator, we can relate error of the two-stage approach to that of an oracle regression of the true $f$ on covariates $X$ in the test data. We begin with the fundamental stability condition required for the second stage regression.

\begin{definition}[Assumption B.1 in \citet{rambachan2022}]\label{def:L2-stability}
    Let $D^n = (O_{01}, \ldots, O_{0n})$ and $O^n = (O_1,\ldots, O_n)$ be training and test sets, respectively, where both are random random samples from $\mathbb{P}$ such that $D^n \independent O^n$. Let $\widehat{f}(o) \equiv \widehat{f}(o ; D^n)$ be estimated with $D^n$, and $\widehat{b}(X) = \mathbb{E}(\widehat{f}(O) - f(O) \mid X, D^n)$ be the bias of $\widehat{f}$ at $X$, conditional on the training data. Let $\widehat{\mathbb{E}}_n(\, \cdot  \mid X = x)$ be a generic regression procedure that regresses functions of $O$ on $X$ in $O^n$, evaluated at $X = x$, and let $\rho$ be a metric on the space of functions of $O$. We say $\widehat{\mathbb{E}}_n$ is $L_2(\mathbb{P})$-\textit{stable for} $f$ \textit{with respect to} $\rho$ if
    \[\frac{\int \left\{\widehat{\mathbb{E}}_n(\widehat{f}(O) \mid X = x) - \widehat{\mathbb{E}}_n(f(O) \mid X = x) - \widehat{\mathbb{E}}_n(\widehat{b}(X) \mid X = x)\right\}^2\, d\mathbb{P}(x)}{\mathbb{E}\left(\int \left\{\widehat{\mathbb{E}}_n(f(O) \mid X = x) - \mathbb{E}(f(O) \mid X = x)\right\}^2 \, d\mathbb{P}(x)\right)} \overset{\mathbb{P}}{\to} 0,\]
    whenever $\rho(\widehat{f}, f) \overset{\mathbb{P}}{\to} 0$, where $\overset{\mathbb{P}}{\to}$ denotes convergence in probability under $\mathbb{P}$.
\end{definition}

While abstract, the following result shows that a large class of regression methods satisfy the stability result in Definition~\ref{def:L2-stability}. Namely, any linear smoother is $L_2(\mathbb{P})$-stable, and this includes linear regression, local polynomials, series estimators, smoothing splines, among many others.

\begin{proposition}[Proposition B.1 in \citet{rambachan2022}]
    Let $X^n = (X_1, \ldots, X_n)$ be covariates associated with test data $O^n = (O_1, \ldots, O_n)$. Consider a linear smoother of the form $\widehat{m}(x) = \widehat{\mathbb{E}}_n(\widehat{f}(O) \mid X = x) = \sum_{i = 1}^n w_i(x; X^n)\widehat{f}(O_i)$, and define
    \[\lVert h \rVert_{w^2} = \sum_{i=1}^n \frac{\lVert w_i(\, \cdot \,; X^n) \rVert^2}{\sum_{j=1}^n \lVert w_j(\, \cdot \,; X^n) \rVert^2} \int h(o)^2\, d\mathbb{P}(o \mid X_i),\]
    for any function $h$. Then, defining $\rho(\widehat{f}, f) = \lVert \widehat{f} - f\rVert_{w^2}$, $\widehat{\mathbb{E}}_n$ is $L_2(\mathbb{P})$-stable for $f$ with respect to $\rho$ if $\frac{1}{\lVert \sigma \rVert_{w^2}} = O_{\mathbb{P}}(1)$, where $\sigma(X)^2 = \mathrm{Var}(f(O) \mid X)$.
\end{proposition}

The main benefit of $L_2$-stability is that it allows a very simple way to characterize the rate of convergence of a two-stage pseudo-outcome regression estimator, compared to the oracle rate. The formal result is as follows.

\begin{lemma}[Lemma B.1 in \citet{rambachan2022}]\label{lemma:pseudo-oracle}
    Under the setup of Definition~\ref{def:L2-stability}, define $m(x) = \mathbb{E}(f(O) \mid X = x)$, $\widetilde{m}(x) = \widehat{\mathbb{E}}_n(f(O) \mid X = x)$, $\widehat{m}(x) = \widehat{\mathbb{E}}_n(\widehat{f}(O) \mid X = x)$, and $\widetilde{b}(x) = \widehat{\mathbb{E}}_n(\widehat{b}(X) \mid X = x)$. If
    \begin{enumerate}[(i)]
        \item $\widehat{\mathbb{E}}_n$ is $L_2(\mathbb{P})$-stable for $f$ with respect to $\rho$,
        \item $\rho(\widehat{f}, f) \overset{\mathbb{P}}{\to} 0$,
    \end{enumerate}
    then $\lVert \widehat{m} - \widetilde{m}\rVert = \lVert \widetilde{b} \rVert + o_{\mathbb{P}}(R_n^*)$, where $\{R_n^*\}^2 = \mathbb{E}\left(\lVert \widetilde{m} - m\rVert^2\right)$  is the squared oracle $L_2(\mathbb{P})$ risk. In particular, if $\lVert \widetilde{b} \rVert = o_{\mathbb{P}}(R_n^*)$, then $\lVert\widehat{m} - m\rVert = \lVert \widetilde{m} - m\rVert + o_{\mathbb{P}}(R_n^*)$ i.e., $\widehat{m}$ is oracle efficient.
\end{lemma}

\begin{remark}
    In the version of their paper accessible at the time of writing this manuscript, \citet{rambachan2022} had a typo in the above result, in which $\{R_n^*\}^2$ was written as $\{\mathbb{E}(\lVert \widetilde{m} - m \rVert)\}^2$. We have corrected this here to $\{R_n^*\}^2 = \mathbb{E}\left(\lVert \widetilde{m} - m\rVert^2\right)$.
\end{remark}

As a last result in this appendix, we present a way to relate the bias $\widehat{b}(x)$ to its smoothed version $\widetilde{b}(x) = \widehat{\mathbb{E}}_n(\widehat{b}(X) \mid X=x)$, in the case of linear smoothers.

\begin{proposition}[Proposition 2 in \citet{kennedy2023}]
    If $\widehat{b}(x) = \widehat{b}_1(x) \widehat{b}_2(x)$, and $\widehat{\mathbb{E}}_n$ is a linear smoother as above, such that $\sum_{i=1}^n |w_i(x; X^n)| = O_{\mathbb{P}}(c_{n, x})$, then
    \[\widetilde{b}(x) = O_{\mathbb{P}}\left(c_{n,x} \lVert \widehat{b}_1\rVert_{w_x,p} \lVert \widehat{b}_2\rVert_{w_x,1}\right),\]
    for any $p, q > 1$ such that $\frac{1}{p} + \frac{1}{q} = 1$, where $\lVert h \rVert_{w_x, p}^p = \sum_{i=1}^n \frac{|w_i(x; X^n)|}{\sum_{j=1}^n |w_j(x; X^n)|} |h(O_i)|^p$.
\end{proposition}

Since for many linear smoothers (e.g., local polynomial estimators), $\sum_{i=1}^n |w_i(x; X^n)| \leq C$ for some universal constant $C$, this result shows that the smoothed bias $\widehat{b}$ can be expressed as a function of the weighted-norms of the components of the raw bias $\widehat{b}$ itself.


\section{Policies Based on a Subset of Covariates}\label{app:subset}

We consider, in this appendix, the case where policies are restricted to depend only on a subset of baseline covariates / confounders $W \subseteq X$. In general we will say $W = g(X)$ for some fixed and known coarsening function $g$, and let $\mathcal{W}$ be the support of $W$. As a simple example, we may be concerned that real-world decision makers (e.g., physicians) will only have access to a subset of covariates (e.g., patient characteristics) compared to those that happen to be measured in a study context. Alternatively, $X$ may be high-dimensional, and for interpretability or ease of use one might only want to consider treatment rules that depend on some lower dimensional summary of $X$. We will consider two possibilities: (i) the contact rule $\Delta$ and policy $h$ (to be applied on $\Delta$-selected subgroup) both may only depend on $W$, and (ii) the contact rule $\Delta$ may only depend on $W$, but the ensuing policy $h$ may use all of $X$. The second possibility may arise if one has limited information with which to \textit{select} a group to intervene on, but once these subjects are contacted one can measure all covariates $X$ prior to treatment selection.

\subsection{Restrictions on contact rules and policies}\label{app:rest-both}

Proceeding as in Section~\ref{sec:framework}, we consider contact rules $\Delta: \mathcal{W}\to [0,1]$, policies $h : \mathcal{W} \to \{0,1\}$, and treatment rules $d(\Delta, h) = H_{\Delta} h + (1 - H_{\Delta})A$, where $H_{\Delta} \overset{d}=\mathrm{Bernoulli}(\Delta(W))$. That is, $d(\Delta, h)$ sets treatment according to $h(W)$ with probability $\Delta(W)$, and leaves the natural value $A$ with probability $1 - \Delta(W)$. Importantly, we will make the same identifying Assumptions~\ref{ass:consistency}--\ref{ass:NUC}, where $X$ are considered to be \textit{all} the confounders that we need to adjust for, i.e., $A \independent Y(a) \mid X$. If $W$ alone is considered to be a sufficient adjustment set, then all the identification and estimation results go through by replacing $X$ with $W$. Otherwise, we will see that the distinction between $W$ and $X$ is necessary, albeit subtle.

First, we note that the result of Proposition~\ref{prop:ident} applies unchanged to treatment rules based on $W$, since these are a subset of rules that use all of $X$. Concretely, under Assumptions~\ref{ass:consistency}--\ref{ass:NUC}, $\mathbb{E}(Y(d(\Delta,h))) = \mathbb{E}(\Delta \tau (h - \pi) + Y)$, where $\tau(X) = \mu_1(X) - \mu_0(X)$ and $\pi(X) = \mathbb{P}[A = 1 \mid X]$ as before. In order to generalize Proposition~\ref{prop:optimal}, we will need some additional notation: let $\tau_w(W) = \mathbb{E}(\tau \mid W)$, $\xi_w(W) = \mathbb{E}(\tau \pi \mid W)$, and define $h_w^* = \mathds{1}(\tau_w > 0)$. Then, according to the following result, the appropriate priority score for maximizing expected outcomes under constraints on intervention capacity is $\beta_w \coloneqq \tau_w h_w^* - \xi_w$.

\begin{proposition}\label{prop:optimal-subset}
Consider the optimization problem
    \begin{align}
    \begin{split}\label{eq:opt-policy-subset}
        \max_{h, \Delta} \; & \mathbb{E}[Y(d(\Delta, h))] \\
        \text{subject to } \; & \mathbb{E}(\Delta(W)) \leq \delta 
    \end{split}
    \end{align}
    where $h$ and $\Delta$ are restricted to have domain $\mathcal{W}$. Under Assumptions~\ref{ass:consistency}--\ref{ass:NUC}, and assuming the $(1 - \delta)$-quantile of $\beta_w$ is unique, denoted $r_{1 - \delta}$, the contact rule $\Delta_{w, \delta}^*(W) = \mathds{1}(\beta_w(W) > r_{1 - \delta})$ and $h_w^*$ jointly solve~\eqref{eq:opt-policy-subset}, achieving maximal value $\mathbb{E}[Y(d(\Delta_{w, \delta}^*, h_w^*))] = \mathbb{E}(\Delta_{w, \delta}^* \beta_w + Y) \eqqcolon V_{w, \delta}$.
\end{proposition}

\begin{proof}
  As in the proof of Proposition~\ref{prop:optimal}, we first show that for any fixed $\Delta: \mathcal{W} \to [0,1]$, the policy $h_w^* = \mathds{1}(\tau_w > 0)$ optimizes $\mathbb{E}[Y(d(\Delta, h))]$ over all possible policies $h: \mathcal{W} \to \{0,1\}$. To see this, observe by Proposition~\ref{prop:ident} that for arbitrary policy $h: \mathcal{W} \to \{0,1\}$,
  \begin{align*}
      \mathbb{E}[Y(d(\Delta, h_w^*))] - \mathbb{E}[Y(d(\Delta, h))] &= \mathbb{E}(\Delta(W) \tau(X) (h_w^*(W) - h(W)) \\
      &= \mathbb{E}(\Delta(W) \tau_w(W) (h_w^*(W) - h(W)),
  \end{align*}
  by iterated expectations given $W$.
  But note that by definition of $h_w^*$,
  \[\tau_w (h_w^* - h) = h_w^* \tau_w(1 - h) - (1 - h_w^*)\tau_w h= |\tau_w| \{h_w^*(1 - h) + (1 - h_w^*)h\} \geq 0, \]
  from which it immediately follows that $\mathbb{E}[Y(d(\Delta, h_w^*))] \geq \mathbb{E}[Y(d(\Delta, h))]$.

  The remaining goal is to maximize, again by Proposition~\ref{prop:ident} and iterated expectations,
  \begin{align*}
      \mathbb{E}(\Delta(W)\tau(X)(h_w^*(W) - \pi(X))) = \mathbb{E}(\Delta(W)\{\tau_w(W)h_w^*(W) - \xi_w(W)\}) =
      \mathbb{E}(\Delta\beta_w)
  \end{align*} 
  over $\left\{\Delta : \mathcal{W} \to [0,1] \mid \mathbb{E}(\Delta(W)) \leq \delta\right\}$.
  This takes exactly the same form as the optimization problem we solved in the proof of Proposition~\ref{prop:optimal}, and the result follows immediately by an identical argument.
\end{proof}

With identification and characterization of optimal rules out of the way, we now briefly show how robust and efficient estimation can proceed, analogously to the approaches laid out in Section~\ref{sec:estimation}. The crucial component will be a new robust pseudo-outcome:
\begin{equation}\label{eq:beta-pseudo-subset}
\phi_w(O;\mathbb{P}) = \left(h_w^* - \pi\right)\left\{\frac{A}{\pi} - \frac{1 - A}{1 - \pi}\right\}\left(Y - \mu_A\right) + \tau\left(h_w^* - A\right).
\end{equation}
The function $\phi_w$ satisfies $\mathbb{E}(\phi_w(O; \mathbb{P}) \mid W) = \tau_w h_w^* - \xi_w = \beta_w$, and has a second-order bias property, similar to that for $\phi$ proved in Lemma~\ref{lemma:pseudo-bias}.

\begin{lemma}\label{lemma:pseudo-bias-subset}
Let $\widetilde{\mathbb{P}}$ be an alternative fixed distribution on $O$. Then the function $\phi_w$ defined in~\eqref{eq:beta-pseudo-subset} satisfies the following conditional bias decomposition:
\begin{align*}
    &\mathbb{E}\left(\phi_w(O; \widetilde{\mathbb{P}}) - \phi_w(O;\mathbb{P}) \mid W\right) \\
    &= \mathbb{E}\bigg(\left(\widetilde{h}_w^* - \widetilde{\pi}\right) \sum_{a=0}^1
      \frac{\{\widetilde{\pi} - \pi\}
      \{\widetilde{\mu}_a - \mu_a\}}{a\widetilde{\pi} + (1 - a)(1 - \widetilde{\pi})} \,\, \bigg| \, \,  W\bigg) + \mathbb{E}(\{\widetilde{\tau} - \tau\}\{\widetilde{\pi} - \pi\} \mid W)
      + \{\widetilde{h}_w^* - h_w^*\}\tau_w,
\end{align*}
where $(\widetilde{\mu}_0, \widetilde{\mu}_1, \widetilde{\pi}, \widetilde{\tau}, \widetilde{h}_w^*)$ represent the corresponding nuisance functions under $\widetilde{\mathbb{P}}$.
\end{lemma}

\begin{proof}
    The result follows from the argument in the proof of Lemma~\ref{lemma:pseudo-bias}, and applying iterated expectations.
\end{proof}

One can leverage the robust pseudo-outcome $\phi_w$ to estimate all the relevant quantities in a robust and efficient manner. For instance, one can construct a ``doubly robust'' learner of $\beta_w$ by regressing $\phi_w(O;\widehat{\mathbb{P}})$ on $W$. We will, for brevity, provide details only for estimating the optimal constrained $W$-based value $V_{w, \delta}$, defined in Proposition~\ref{prop:optimal-subset}.

Following the development in Section~\ref{sec:estimation-value}, let $\widehat{\tau}_w$ be estimated in an arbitrary manner, e.g., regressing $\widehat{\tau}$ on $W$. From this, we can define the plug-in estimator $\widehat{h}_w^* = \mathds{1}(\widehat{\tau}_w > 0)$. For a given estimator $\widehat{\beta}_w$ constructed from training data $D^n$ (e.g., a plug-in $\widehat{\tau}_w\widehat{h}_w^* - \widehat{\xi}_w$, or else a doubly robust learner), let $\widehat{r}_{1 - \delta}$ solve $ \mathbb{P}_n\left[\widehat{\beta}_w(W) > \widehat{r}_{1 - \delta}\right] = \delta$ up to $o_{\mathbb{P}}(n^{-1/2})$ error, and define the plug-in optimal constrained contact rule as $\widehat{\Delta}_{w, \delta}^* = \mathds{1}\left(\widehat{\beta}_w > \widehat{r}_{1 - \delta}\right)$. The proposed estimator of $V_{w,\delta}$ is then given by
\[\widehat{V}_{w,\delta} = \mathbb{P}_n\left[\widehat{\Delta}_{w,\delta}^*\phi_w(O; \widehat{\mathbb{P}}) + Y\right],\]
where $\phi_w(O; \widehat{\mathbb{P}})$ is obtained from equation~\eqref{eq:beta-pseudo-subset}, plugging in $(\widehat{\pi}, \widehat{\mu}_0, \widehat{\mu}_1)$ (and the derived nuisance estimates $\widehat{\tau}, \widehat{h}_w^*$).

We can develop the asymptotic properties of $\widehat{V}_{w, \delta}$ (in a result akin to Theorem~\ref{thm:est-value}), but first we need an appropriate margin condition (similar to Assumption~\ref{ass:margin}). We next present the relevant assumption, along with a formal result.

\begin{assumption}\label{ass:margin-subset}
    For some $a_w, b_w > 0$,
    \[\mathbb{P}\left[|\tau_w(W)| \leq t\right]
  \lesssim t^{a_w}, \text{ and }\mathbb{P}[|\beta_w(W) - r_{1 - \delta}
  | \leq t] \lesssim t^{b_w} \text{ for all } t \geq 0.\]
\end{assumption}

\begin{theorem}\label{thm:est-value-subset}
  Assume that $\lVert \widehat{\mu}_0 -\mu_0\rVert + \lVert \widehat{\mu}_1 -\mu_1\rVert + \lVert \widehat{\pi} -\pi\rVert + \lVert \widehat{\beta}_w -\beta_w\rVert + \lVert \widehat{\tau}_w -\tau_w\rVert + |\widehat{r}_{1 - \delta} - r_{1 - \delta}| = o_{\mathbb{P}}(1)$. Moreover, assume that there exists $\epsilon > 0, M > 0$ such that $\mathbb{P}[\epsilon \leq \pi \leq 1 - \epsilon] = \mathbb{P}[\epsilon \leq \widehat{\pi} \leq 1 - \epsilon] = 1$, $\mathbb{P}[|Y| \leq M] = 1$. Then, under Assumption~\ref{ass:margin-subset},
  \[\widehat{V}_{w,\delta} - V_{w,\delta} =  O_{\mathbb{P}}\left(\frac{1}{\sqrt{n}} + R_{1,n} + R_{2,w,n} + R_{3,w,n}\right),\]
  where 
  \[R_{1,n} = \left\lVert \widehat{\pi} - \pi \right\rVert
    \left(\lVert \widehat{\mu}_0 -\mu_0\rVert + \lVert \widehat{\mu}_1
      -\mu_1\rVert\right),\]
  and
  \[R_{2,w,n} = \left\lVert \widehat{\tau}_w - \tau_w
      \right\rVert_{\infty}^{1 + a_w}, \ R_{3,w,n} = \left(\lVert \widehat{\beta}_w - \beta_w
      \rVert_{\infty} + |\widehat{r}_{1 - \delta} - r_{1 -
        \delta}|\right)^{1 + b_w}.\] 
  If, in addition, $R_{1,n} + R_{2,w,n} + R_{3,w,n} = o_{\mathbb{P}}(n^{-1/2})$, then \[\sqrt{n}(\widehat{V}_{w,\delta} - V_{w,\delta}) \overset{d}{\to} \mathcal{N}\left(0, \sigma_w(\delta)^2\right),\]
  where $\sigma_w(\delta)^2 = \mathrm{Var}(\Delta_{w,\delta}^* \left\{\phi_w(O;\mathbb{P}) - r_{1-\delta}\right\} + Y)$.
\end{theorem}

The proof of Theorem~\ref{thm:est-value-subset} follows the exact reasoning as the proof of Theorem~\ref{thm:est-value}. Moreover, similar comments to those that follow Theorem~\ref{thm:est-value} in Section~\ref{sec:estimation-value} are relevant again here. For instance, in the case that $R_{1,n} + R_{2,w,n} + R_{3,w,n} = o_{\mathbb{P}}(n^{-1/2})$, one can obtain asymptotically valid inference via simple Wald-based confidence intervals for $V_{w,\delta}$:
  $\widehat{V}_{w,\delta} \pm z_{1- \alpha/2}\frac{\widehat{\sigma}_w(\delta)}{\sqrt{n}}$, where
  \[\widehat{\sigma}_w(\delta)^2 = \mathbb{P}_n\left[\left(\widehat{\Delta}_{w,\delta}^* \left\{\phi_w(O;\widehat{\mathbb{P}}) - \widehat{r}_{1-\delta}\right\} + Y - \left\{\widehat{V}_{w,\delta} - \delta \widehat{r}_{1 - \delta}\right\}\right)^2\right].\]

\subsection{Restriction on contact rule only}

Now we consider contact rules $\Delta: \mathcal{W}\to [0,1]$, but unrestricted policies $h : \mathcal{X} \to \{0,1\}$. As before, we are interested in treatment rules $d(\Delta, h) = H_{\Delta} h + (1 - H_{\Delta})A$, where $H_{\Delta} \overset{d}=\mathrm{Bernoulli}(\Delta(W))$, and we work under Assumptions~\ref{ass:consistency}--\ref{ass:NUC}, such that $X$ are considered to be all the confounders for which we need to adjust.

For brevity, we provide only the formal results and value estimators for these treatment rules, as the arguments follow exactly as in preceding results. In this case, the appropriate priority score for targeted intervention is the conditional mean of the CPB given $W$: $\check{\beta}_w(W) \coloneqq \mathbb{E}(\beta \mid W)$.

\begin{proposition}\label{prop:optimal-subcont}
Consider the optimization problem
    \begin{align}
    \begin{split}\label{eq:opt-policy-subcont}
        \max_{h, \Delta} \; & \mathbb{E}[Y(d(\Delta, h))] \\
        \text{subject to } \; & \mathbb{E}(\Delta(W)) \leq \delta 
    \end{split}
    \end{align}
    where $\Delta$ is restricted to have domain $\mathcal{W}$. Under Assumptions~\ref{ass:consistency}--\ref{ass:NUC}, and assuming the $(1 - \delta)$-quantile of $\check{\beta}_w$ is unique, denoted $s_{1 - \delta}$, the contact rule $\check{\Delta}_{w, \delta}^*(W) = \mathds{1}(\check{\beta}_w(W) > s_{1 - \delta})$ and $h^*$ jointly solve~\eqref{eq:opt-policy-subcont}, achieving maximal value $\mathbb{E}[Y(d(\check{\Delta}_{w, \delta}^*, h^*))] = \mathbb{E}(\check{\Delta}_{w, \delta}^*\beta + Y) \eqqcolon \check{V}_{w, \delta}$.
\end{proposition}

For a given estimator $\widehat{\check{\beta}}_w$ constructed from training data $D^n$ (e.g., regressing $\widehat{\beta}$ on $W$, or else a doubly robust learner regressing $\phi(O; \widehat{\mathbb{P}})$ on $W$), let $\widehat{s}_{1 - \delta}$ solve $ \mathbb{P}_n\left[\widehat{\check{\beta}}_w(W) > \widehat{s}_{1 - \delta}\right] = \delta$ up to $o_{\mathbb{P}}(n^{-1/2})$ error, and define the plug-in optimal constrained contact rule as $\widehat{\check{\Delta}}_{w, \delta}^* = \mathds{1}\left(\widehat{\check{\beta}}_w > \widehat{s}_{1 - \delta}\right)$. The proposed estimator of $\check{V}_{w,\delta}$ is then given by
\[\widehat{\check{V}}_{w,\delta} = \mathbb{P}_n\left[\widehat{\check{\Delta}}_{w,\delta}^*\phi(O; \widehat{\mathbb{P}}) + Y\right],\]
where $\phi(O; \widehat{\mathbb{P}})$ is obtained from equation~\eqref{eq:beta-pseudo}, plugging in $(\widehat{\pi}, \widehat{\mu}_0, \widehat{\mu}_1)$ (and the derived nuisance estimates $\widehat{\tau}, \widehat{h}^*$).

\begin{assumption}\label{ass:margin-subcont}
    For some $\check{b}_w > 0$,
    \[\mathbb{P}[|\check{\beta}_w(W) - s_{1 - \delta}
  | \leq t] \lesssim t^{\check{b}_w} \text{ for all } t \geq 0.\]
\end{assumption}

\begin{theorem}\label{thm:est-value-subcont}
  Assume that $\lVert \widehat{\mu}_0 -\mu_0\rVert + \lVert \widehat{\mu}_1 -\mu_1\rVert + \lVert \widehat{\pi} -\pi\rVert + \lVert \widehat{\check{\beta}}_w -\check{\beta}_w\rVert + |\widehat{s}_{1 - \delta} - s_{1 - \delta}| = o_{\mathbb{P}}(1)$. Moreover, assume that there exists $\epsilon > 0, M > 0$ such that $\mathbb{P}[\epsilon \leq \pi \leq 1 - \epsilon] = \mathbb{P}[\epsilon \leq \widehat{\pi} \leq 1 - \epsilon] = 1$, $\mathbb{P}[|Y| \leq M] = 1$. Then, under Assumptions~\ref{ass:margin} and \ref{ass:margin-subcont},
  \[\widehat{\check{V}}_{w,\delta} - \check{V}_{w,\delta} =  O_{\mathbb{P}}\left(\frac{1}{\sqrt{n}} + R_{1,n} + R_{2,n} + \check{R}_{3,w,n}\right),\]
  where 
  \[R_{1,n} = \left\lVert \widehat{\pi} - \pi \right\rVert
    \left(\lVert \widehat{\mu}_0 -\mu_0\rVert + \lVert \widehat{\mu}_1
      -\mu_1\rVert\right),\]
  and
  \[R_{2,n} = \left\lVert \widehat{\tau} - \tau
      \right\rVert_{\infty}^{1 + a}, \ \check{R}_{3,w,n} = \left(\lVert \widehat{\check{\beta}}_w - \check{\beta}_w
      \rVert_{\infty} + |\widehat{s}_{1 - \delta} - s_{1 -
        \delta}|\right)^{1 + \check{b}_w}.\] 
  If, in addition, $R_{1,n} + R_{2,n} + \check{R}_{3,w,n} = o_{\mathbb{P}}(n^{-1/2})$, then \[\sqrt{n}(\widehat{\check{V}}_{w,\delta} - \check{V}_{w,\delta}) \overset{d}{\to} \mathcal{N}\left(0, \check{\sigma}_w(\delta)^2\right),\]
  where $\check{\sigma}_w(\delta)^2 = \mathrm{Var}(\check{\Delta}_{w,\delta}^* \left\{\phi(O;\mathbb{P}) - s_{1-\delta}\right\} + Y)$.
  
\end{theorem}

Another observation is that the treatment rules considered in this subsection---in contrast to those in Section~\ref{app:rest-both}---are directly comparable to those that are not restricted and can depend fully on $X$. In particular, when $\delta = 1$, $d(\check{\Delta}_{w, 1}^*, h^*) = d(\Delta_1^*, h^*)= h^*$, so that $\check{V}_{w, 1} = V_1 = \mathbb{E}(Y(h^*))$. Consequently, we can consider a $W$-restricted AUPBC metric, which takes the following integral form:
\begin{equation}\label{eq:AUPBC-subset}
    \mathcal{A}_w = \int_0^1 \mathbb{E}\left((\check{\Delta}_{w,\delta}^* - \delta)\beta\right) \, d \delta.
\end{equation}
Its normalized version is $\overline{\mathcal{A}}_w = 2\mathcal{A}_w / \mathbb{E}(\beta) \in [0,1]$, and estimation and inference for these quantities follows similarly to their counterparts $\mathcal{A}$ and $\overline{\mathcal{A}}$. Note that, by construction, $\overline{\mathcal{A}}_w \leq \overline{\mathcal{A}}$, and the gap between these two summarizes how much more the information in $X$, compared to $W$ only, helps in targeting optimal interventions under budget constraints.

\section{Auxiliary \& Technical Lemmas}
In this appendix, we present several auxiliary results that are used in the proofs of other results.

\begin{lemma}[Lemma 1 in \citet{kennedy2020b}] \label{lemma:diff}
    Let $x, z \in \mathbb{R}$ be arbitrary. Then
    \[|\mathds{1}(x > 0) - \mathds{1}(z>0)| \leq \mathds{1}(\max{\{|x|, |z|\}}
      \leq |x - z|).\]
  \end{lemma}
  \begin{proof}
    If $|\mathds{1}(x > 0) - \mathds{1}(z>0)| = 0$, the equality holds
    trivially. Otherwise, either $x > 0$ and $z \leq 0$, meaning
    $|x| = x \leq x + |z| = x - z = |x - z|$, or $x \leq 0$ and $z > 0$,
    meaning $|x| = -x \leq z - x = |x - z|$. Symmetrically,
    $|z| \leq |x - z|$ in these two cases, so that
    $\max{\{|x|, |z|\}} \leq |x - z|$ when
    $|\mathds{1}(x > 0) - \mathds{1}(z>0)| = 1$.
  \end{proof}

\begin{lemma}[Lemma 2 in \citet{kennedy2020b}] \label{lemma:emp-process}
    Let $\widehat{f}(o)$ be a function estimated from training data $D^n = (O_{01}, \ldots, O_{0n})$, and let $\mathbb{P}_n$ be the empirical measure on $O^n = (O_1, \ldots, O_n)$ where $D^n$ and $O^n$ are iid samples from $\mathbb{P}$ with $O^n \independent D^n$. Write $\mathbb{P}(h) = \int h(o)\,d\mathbb{P}(o)$ for the mean of any function $h$ (possibly data-dependent) over a new observation. Then
    \begin{equation}\label{eq:kennedy-lemma}
    (\mathbb{P}_n - \mathbb{P})(\widehat{f} - f) = O_{\mathbb{P}}\left(\frac{\lVert \widehat{f} - f\rVert}{\sqrt{n}}\right),
    \end{equation}
    where $\lVert h \rVert^2 = \int h(o)^2 \, d\mathbb{P}(o)$, for any $h$.
\end{lemma}

\begin{proof}
    Note that
    \[\mathbb{E}(\mathbb{P}_n[\widehat{f} - f ] \mid D^n) = \mathbb{E}(\widehat{f}(O) - f(O) \mid D^n) = \mathbb{P}(\widehat{f} - f),\]
    by identical distribution, $O^n \independent D^n$, and by definition of the operator $\mathbb{P}$. Moreover,
    \[\mathrm{Var}((\mathbb{P}_n - \mathbb{P})[\widehat{f} - f ]\mid D^n) = \mathrm{Var}(\mathbb{P}_n[\widehat{f} - f ]\mid D^n) = \frac{1}{n}\mathrm{Var}(\widehat{f} - f \mid D^n) \leq \frac{1}{n}\lVert \widehat{f} - f\rVert^2,\]
    by independence and identical distribution, and using the fact that $\mathrm{Var}(W) \leq \mathbb{E}(W^2)$ for any $W$. Thus, $\mathbb{E}(\{(\mathbb{P}_n - \mathbb{P})[\widehat{f} - f ]\}^2 \mid D^n) \leq \frac{1}{n} \lVert \widehat{f} - f\rVert^2$, and so
    \begin{align*}
    \mathbb{P}\left[\sqrt{n}\frac{(\mathbb{P}_n - \mathbb{P})(\widehat{f} - f)}{\lVert \widehat{f} - f\rVert} \geq t\right] 
    = \mathbb{E}\left(\mathbb{P}\left[\sqrt{n}\frac{(\mathbb{P}_n - \mathbb{P})(\widehat{f} - f)}{\lVert \widehat{f} - f\rVert} \geq t \, \bigg| \, D^n\right]\right) \leq \frac{1}{t^2},
    \end{align*}
    by applying Markov's inequality conditional on $D^n$. For any $\epsilon > 0$, we can choose $t = \frac{1}{\sqrt{\epsilon}}$ to bound this probability by $\epsilon$, which yields~\eqref{eq:kennedy-lemma}.
\end{proof}

    \begin{lemma} \label{lemma:convbound}
    Suppose that for a given sequence $X_n$, one can find for any
    $\epsilon > 0$ another sequence $Z_n^{(\epsilon)} \geq 0$ such that
    $|X_n| \leq \epsilon + Z_n^{(\epsilon)}$ and
    $Z_n^{(\epsilon)} = o_{\mathbb{P}}(1)$. Then $X_n = o_{\mathbb{P}}(1)$.
  \end{lemma}

    \begin{proof}
    Fixing $\epsilon > 0$, consider a non-negative sequence
    $Z_n^{(\epsilon / 2)} = o_{\mathbb{P}}(1)$ satisfying
    $|X_n| \leq \epsilon / 2 + Z_n^{(\epsilon / 2)}$. Then
  \[{\mathbb{P}}[|X_n| > \epsilon] \leq {\mathbb{P}}[\epsilon / 2 + Z_n^{(\epsilon / 2)} >
    \epsilon] = {\mathbb{P}}[Z_n^{(\epsilon / 2)} > \epsilon / 2] \to 0 \text{ as
    } n \to \infty,\] since $Z_n^{(\epsilon / 2)} = o_{\mathbb{P}}(1)$, thus proving the result.
\end{proof}

\begin{lemma}\label{lemma:quantile-emp-process}
Under the conditions of Theorem~\ref{thm:est-value},
    \[\left(\mathbb{P}_n - \mathbb{P}\right)\left(\left\{\widehat{\Delta}_{\delta}^* - \Delta_{\delta}^*\right\}\phi(O; \mathbb{P})\right) = o_{\mathbb{P}}(n^{-1/2}),\]
    and similarly 
    \[\left(\mathbb{P}_n - \mathbb{P}\right)\left(\left\{\widehat{\Delta}_{\delta}^* - \Delta_{\delta}^*\right\}\phi(O; \widehat{\mathbb{P}})\right) = o_{\mathbb{P}}(n^{-1/2}), \ \left(\mathbb{P}_n - \mathbb{P}\right)\left(\widehat{\Delta}_{\delta}^* - \Delta_{\delta}^*\right) = o_{\mathbb{P}}(n^{-1/2}).\]
\end{lemma}

\begin{proof}
    We show the first result, as the other two are similar. Observe that
    \begin{align*}
    \left(\mathbb{P}_n - \mathbb{P}\right)\left(\left\{\widehat{\Delta}_{\delta}^* - \Delta_{\delta}^*\right\}\phi(O; \mathbb{P})\right)
    &= \left(\mathbb{P}_n - \mathbb{P}\right)\left(\left\{\mathds{1}(\widehat{\beta} > \widehat{q}_{1 - \delta}) - \mathds{1}(\beta > q_{1 - \delta})\right\}\phi(O; \mathbb{P})\right) \\
    &= \left(\mathbb{P}_n - \mathbb{P}\right)\left(\left\{\mathds{1}(\widehat{\beta} > q_{1 - \delta}) - \mathds{1}(\beta > q_{1 - \delta})\right\}\phi(O; \mathbb{P})\right) \\
    & \quad \quad + \left(\mathbb{P}_n - \mathbb{P}\right)\left(\left\{\mathds{1}(\widehat{\beta} > \widehat{q}_{1 - \delta}) - \mathds{1}(\widehat{\beta} > q_{1 - \delta})\right\}\phi(O; \mathbb{P})\right)
    \end{align*}

    For the first of these two terms, we use consistency and the fact that we used sample splitting. The second term requires more careful analysis since $\widehat{q}_{1 - \delta}$ depends on $O^n$.

    For the first term, note that for any $t > 0$,
    \begin{align*}
        \left\lVert \left\{\mathds{1}(\widehat{\beta} > q_{1 - \delta}) - \mathds{1}(\beta > q_{1 - \delta})\right\}\phi(O; \mathbb{P})\right\rVert^2
        &= \mathbb{P}\left(\left\{\mathds{1}(\widehat{\beta} > q_{1 - \delta}) - \mathds{1}(\beta > q_{1 - \delta})\right\}^2\phi(O; \mathbb{P})^2\right) \\
        &\lesssim \mathbb{P}\left(\left|\mathds{1}(\widehat{\beta} > q_{1 - \delta}) - \mathds{1}(\beta > q_{1 - \delta})\right|\right) \\
        & \leq \mathbb{P}\left[|\beta - q_{1 - \delta}| \leq |\widehat{\beta} - \beta|\right] \\
        & \leq \mathbb{P}\left[|\beta - q_{1 - \delta}| \leq t\right]  + \mathbb{P}\left[  |\widehat{\beta} - \beta| > t\right] \\
        & \lesssim t^b + \frac{\lVert \widehat{\beta} - \beta\rVert}{t}
    \end{align*}
    where the second line follows from boundedness of $\phi(O; \mathbb{P})$, the third by Lemma~\ref{lemma:diff}, the fourth by logic, and the fifth by Assumption~\ref{ass:margin} and Markov's inequality. Thus, the norm on the left-hand side is $o_{\mathbb{P}}(1)$ by consistency of $\widehat{\beta}$ and Lemma~\ref{lemma:convbound}, so that our first term is $o_{\mathbb{P}}(n^{-1/2})$ by Lemma~\ref{lemma:emp-process}.


    For the second term, note that the function class $\left\{\mathds{1}(\overline{\beta}(\, \cdot \,) > b) : b \in \mathbb{R}\right\}$
    is Donsker (in $b$) for any fixed function $\overline{\beta}: \mathcal{X} \to \mathbb{R}$ (see Example 2.5.4 in \citet{vandervaart1996}). Since $\phi(O; \mathbb{P})$ is uniformly bounded, the class $\left\{\mathds{1}(\overline{\beta}(\, \cdot\, ) > b) \phi(\, \cdot \,; \mathbb{P}) : b \in \mathbb{R}\right\}$ is also Donsker (see Example 2.10.10 in \citet{vandervaart1996}). Working similarly as above, we see that for any $t > 0$,
    \begin{align*}
        & \left\lVert \left\{\mathds{1}(\widehat{\beta} > \widehat{q}_{1 - \delta}) - \mathds{1}(\widehat{\beta} > q_{1 - \delta})\right\}\phi(O; \mathbb{P})\right\rVert^2 \\
        &= \mathbb{P}\left(\left\{\mathds{1}(\widehat{\beta} > \widehat{q}_{1 - \delta}) - \mathds{1}(\widehat{\beta} > q_{1 - \delta})\right\}^2\phi(O; \mathbb{P})^2\right) \\
        & \lesssim \mathbb{P}\left(\left|\mathds{1}(\widehat{\beta} > \widehat{q}_{1 - \delta}) - \mathds{1}(\widehat{\beta} > q_{1 - \delta})\right|\right) \\
        & \leq \mathbb{P}\left(|\widehat{\beta} - q_{1 - \delta}| \leq |\widehat{q}_{1 - \delta} - q_{1 - \delta}|\right) \\
        & \leq \mathbb{P}\left(|\beta - q_{1 - \delta}| \leq |\widehat{q}_{1 - \delta} - q_{1 - \delta}| + |\widehat{\beta} - \beta|\right) \\
        & \leq \mathbb{P}\left(|\beta - q_{1 - \delta}| \leq t\right) + \mathbb{P}\left(|\widehat{q}_{1 - \delta} - q_{1 - \delta}| + |\widehat{\beta} - \beta| > t\right) \\
        & \lesssim t^b + \frac{|\widehat{q}_{1 - \delta} - q_{1 - \delta}| + \lVert\widehat{\beta} - \beta\rVert}{t},
    \end{align*}
    where the second line follows from boundedness of $\phi(O; \mathbb{P})$, the third by Lemma~\ref{lemma:diff}, the fourth by the triangle inequality (i.e., $|\beta - q_{1 - \delta}| \leq |\widehat{\beta} - q_{1 - \delta}| + |\widehat{\beta} - \beta|$), the fifth by logic, and the sixth by Assumption~\ref{ass:margin} and Markov's inequality. This shows the norm on the left-hand side is $o_{\mathbb{P}}(1)$ by Lemma~\ref{lemma:convbound}, and consistency of $\widehat{\beta}$ and $\widehat{q}_{1 - \delta}$. Applying Lemma 19.24 of \citet{vandervaart2000} (conditionally on $D^n$) we find that our second term is $o_{\mathbb{P}}(n^{-1/2})$ (conditionally on $D^n$ and hence unconditionally).
\end{proof}

\begin{lemma}\label{lemma:sup-emp-process}
Under the conditions of Theorem~\ref{thm:est-AUPBC},
    \begin{align*}
        &\sup_{\delta \in [0,1]} \left|\widehat{V}_{\delta} - V_{\delta} - (\mathbb{P}_n - \mathbb{P})(\Delta_{\delta}^*\{\phi(O; \mathbb{P}) - q_{1 - \delta}\} + Y)\right| \\
        & = O_{\mathbb{P}}(R_{1,n} + R_{2,n} + R_{3,n}^*) + o_{\mathbb{P}}(n^{-1/2}).
    \end{align*}
\end{lemma}

\begin{proof}
Our proof closely follows that of Theorem 2 in~\citet{bonvini2022}. To begin, see that for any $\delta \in [0,1]$,
    \begin{align*}
        & \widehat{V}_{\delta} - V_{\delta} - (\mathbb{P}_n - \mathbb{P})(\Delta_{\delta}^*\{\phi(O; \mathbb{P}) - q_{1 - \delta}\} + Y) \\
        &= \mathbb{P}_n\left[\widehat{\Delta}_{\delta}^* \phi(O; \widehat{\mathbb{P}})\right] - \mathbb{P}(\Delta_{\delta}^* \phi(O;\mathbb{P})) - (\mathbb{P}_n - \mathbb{P})(\Delta_{\delta}^*\{\phi(O; \mathbb{P}) - q_{1 - \delta}\}) \\
        &= \left(\mathbb{P}_n - \mathbb{P}\right)\left(\widehat{\Delta}_{\delta}^* \phi(O; \widehat{\mathbb{P}}) - \Delta_{\delta}^* \phi(O;\mathbb{P})\right) + \mathbb{P}\left(\widehat{\Delta}_{\delta}^* \phi(O; \widehat{\mathbb{P}}) - \Delta_{\delta}^* \phi(O;\mathbb{P})\right) + \left(\mathbb{P}_n - \mathbb{P}\right)\left(q_{1 - \delta}\Delta_{\delta}^*\right).
    \end{align*}
    Next, see that the middle term can be decomposed via
    \begin{align*}
&\mathbb{P}\left(\widehat{\Delta}_{\delta}^* \phi(O; \widehat{\mathbb{P}}) - \Delta_{\delta}^* \phi(O;\mathbb{P})\right) \\
&= \mathbb{P}\left(\widehat{\Delta}_{\delta}^* \{\phi(O; \widehat{\mathbb{P}}) - \phi(O;\mathbb{P})\} + (\phi(O; \mathbb{P}) - q_{1 - \delta})\{\widehat{\Delta}_{\delta}^* - \Delta_{\delta}^*\} + q_{1 - \delta}\{\widehat{\Delta}_{\delta}^* - \Delta_{\delta}^*\}\right),
    \end{align*}
    and that by construction of $\widehat{q}_{1 - \delta}$,
\begin{align*}
    o_{\mathbb{P}}(n^{-1/2}) &=
\mathbb{P}_n(\widehat{\Delta}_{\delta}^*) - \mathbb{P}(\Delta_{\delta}^*)
= (\mathbb{P}_n - \mathbb{P})(\widehat{\Delta}_{\delta}^* - \Delta_{\delta}^*) + \mathbb{P}(\widehat{\Delta}_{\delta}^* - \Delta_{\delta}^*) + (\mathbb{P}_n - \mathbb{P})(\Delta_{\delta}^*).
\end{align*}
Hence,
\[(\mathbb{P}_n - \mathbb{P})\left(q_{1 - \delta}\{\widehat{\Delta}_{\delta}^* - \Delta_{\delta}^*\}\right) = - q_{1 - \delta}(\mathbb{P}_n - \mathbb{P})(\widehat{\Delta}_{\delta}^* - \Delta_{\delta}^*) - (\mathbb{P}_n - \mathbb{P})(q_{1 - \delta}\Delta_{\delta}^*) + o_{\mathbb{P}}(n^{-1/2}),\]
so that our original decomposition becomes
\begin{align*}
    &\widehat{V}_{\delta} - V_{\delta} - (\mathbb{P}_n - \mathbb{P})(\Delta_{\delta}^*\{\phi(O; \mathbb{P}) - q_{1 - \delta}\} + Y) \\
    &= \left(\mathbb{P}_n - \mathbb{P}\right)\left(\widehat{\Delta}_{\delta}^* \{\phi(O; \widehat{\mathbb{P}}) - q_{1 - \delta}\} - \Delta_{\delta}^* \{\phi(O;\mathbb{P}) - q_{1 - \delta}\}\right) \\
    & \quad \quad + \mathbb{P}\left(\widehat{\Delta}_{\delta}^* \{\phi(O; \widehat{\mathbb{P}}) - \phi(O;\mathbb{P})\} + (\phi(O; \mathbb{P}) - q_{1 - \delta})\{\widehat{\Delta}_{\delta}^* - \Delta_{\delta}^*\}\right) + o_{\mathbb{P}}(n^{-1/2}).
\end{align*}
By our assumption that $\sup_{\delta \in [0,1]}\left|\mathbb{P}_n[\widehat{\Delta}_{\delta}^*] - \delta\right| = o_{\mathbb{P}}(n^{-1/2})$, the third term (i.e., the $o_{\mathbb{P}}(n^{-1/2})$ term) remains negligible after taking a supremum over $\delta \in [0,1]$. Thus, it suffices to show that $\mathcal{E}_n = o_{\mathbb{P}}(n^{-1/2})$ and $\mathcal{B}_n = O_{\mathbb{P}}(R_{1,n} + R_{2,n} + R_{3,n}^*)$, where
\[\mathcal{E}_n = \sup_{\delta \in [0,1]} \left|\left(\mathbb{P}_n - \mathbb{P}\right)\left(\widehat{\Delta}_{\delta}^* \{\phi(O; \widehat{\mathbb{P}}) - q_{1 - \delta}\} - \Delta_{\delta}^* \{\phi(O;\mathbb{P}) - q_{1 - \delta}\}\right)\right|,\]
and
\[\mathcal{B}_n = \sup_{\delta \in [0,1]} \left|\mathbb{P}\left(\widehat{\Delta}_{\delta}^* \{\phi(O; \widehat{\mathbb{P}}) - \phi(O;\mathbb{P})\} + (\phi(O; \mathbb{P}) - q_{1 - \delta})\{\widehat{\Delta}_{\delta}^* - \Delta_{\delta}^*\}\right)\right|.\]

We begin with bounding $\mathcal{E}_n$. To proceed, we will use powerful tools from empirical process theory~\citep{vandervaart1996, vandervaart2000, kosorok2008}. Define the function classes $\mathcal{F}_1 = \{\Delta_{\delta}^* : \delta \in [0,1]\}$ and $\mathcal{F}_2 = \{\phi(O; \mathbb{P}) - q_{1 - \delta} : \delta \in [0,1]\}$, and notice that both classes are uniformly bounded under the assumptions of Theorem~\ref{thm:est-AUPBC}.  The class $\mathcal{F}_1$ is contained in the class of all uniformly bounded, monotone functions, so that by Theorem 2.7.5 in \citet{vandervaart1996}, $\log{N_{[\,]}(\gamma , \mathcal{F}_1, L_2(\mathbb{P}))} \lesssim \frac{1}{\gamma}$, for all $\gamma > 0$. The class $\mathcal{F}_2$ also contains bounded monotone functions $-q_{1 - \delta}$ (plus the single, uniformly bounded function $\phi(O; \mathbb{P})$), so the same bracketing number bound holds. Analogous statements can be made for the classes $\widetilde{\mathcal{F}}_{1,n} = \{\widehat{\Delta}_{\delta}^* : \delta \in [0,1]\}$ and $\widetilde{\mathcal{F}}_{2,n} = \{\phi(O; \widehat{\mathbb{P}}) - q_{1 - \delta} : \delta \in [0,1]\}$, noting again that these are both uniformly bounded under our assumptions. We can apply Lemma 9.24 in \citet{kosorok2008} to find that $\frac{1}{\gamma}$ remains a valid bound (up to constants) on the log-bracketing number, even when combining the function classes through set products and sums. In particular, using the fact that these classes are uniformly bounded, we can conclude that $\log{N_{[\,]}(\gamma , \mathcal{F}_n, L_2(\mathbb{P}))} \lesssim \frac{1}{\gamma}$, for any $\gamma > 0$, where $\mathcal{F}_n = \mathcal{F}_1 \cdot \mathcal{F}_2 - \widetilde{\mathcal{F}}_{1,n} \cdot\widetilde{\mathcal{F}}_{2,n}$.

Now, observe that $\mathcal{E}_n \leq \sup_{f \in \mathcal{F}_n} \left|(\mathbb{P}_n - \mathbb{P})(f)\right|$. Conditioning on the training data $D^n$ (so as to view $(\widehat{\pi}, \widehat{\mu}_0, \widehat{\mu}_1)$ as fixed), and applying Theorem 2.14.2 in \citep{vandervaart1996}, we find that
\[\mathbb{E}\left(\sup_{f \in \mathcal{F}_n} \left|(\mathbb{P}_n - \mathbb{P})(f)\right| \,\middle| \, D^n\right) \lesssim \frac{\lVert F_n\rVert}{\sqrt{n}} \int_0^1 \sqrt{1 + \log{N_{[\,]}(\gamma\lVert F_n\rVert , \mathcal{F}_n, L_2(\mathbb{P}))}} \, d\gamma,\]
where $F_n$ is an envelope for $\mathcal{F}_n$---we will take
\[F_n = \sup_{\delta \in [0,1]} \left|\widehat{\Delta}_{\delta}^* \{\phi(O; \widehat{\mathbb{P}}) - q_{1 - \delta}\} - \Delta_{\delta}^* \{\phi(O; \mathbb{P}) - q_{1 - \delta}\}\right|,\]
which we have assumed is $o_{\mathbb{P}}(1)$. By our bracketing number estimate above, we have
\begin{align*}
    \lVert F_n\rVert \int_0^1 \sqrt{1 + \log{N_{[\,]}(\gamma\lVert F_n\rVert , \mathcal{F}_n, L_2(\mathbb{P}))}} \, d\gamma
    & \lesssim \lVert F_n\rVert \int_0^1 \sqrt{1 + \frac{1}{\gamma \lVert F_n\rVert}} \, d\gamma\\
    & \leq \lVert F_n\rVert + \lVert F_n\rVert^{1/2} \int_0^1 \gamma^{-1/2} \, d\gamma \\
    & = \lVert F_n\rVert + 2\lVert F_n\rVert^{1/2} ,
\end{align*}
where the second inequality used the fact that $\sqrt{a + b} \leq \sqrt{a} + \sqrt{b}$, for all $a, b \geq 0$. 
Since the sequence $\lVert F_n\rVert + 2\lVert F_n\rVert^{1/2}$ is dominated (and thus uniformly integrable), it converges in $L_1(\mathbb{P})$ to 0, so that $\sqrt{n} \mathbb{E}\left(\sup_{f \in \mathcal{F}_n} \left|(\mathbb{P}_n - \mathbb{P})(f)\right|\right) \to 0$. By Markov's inequality, we can conclude that $\sup_{f \in \mathcal{F}_n} \left|(\mathbb{P}_n - \mathbb{P})(f)\right| = o_{\mathbb{P}}(n^{-1/2})$, and thus $\mathcal{E}_n = o_{\mathbb{P}}(n^{-1/2})$.

Finally, we consider bounding $\mathcal{B}_n$. As in the proof of Theorem~\ref{thm:est-value}, by the product bias in Lemma~\ref{lemma:pseudo-bias}, as well as Lemma~\ref{lemma:diff} and Assumption~\ref{ass:margin},
  \begin{align*}
\sup_{\delta \in [0,1]} \left|\mathbb{P}\left(\widehat{\Delta}_{\delta}^*(X)(\phi(O;\widehat{{\mathbb{P}}}) - \phi(O;\mathbb{P}))\right)\right| & \leq \mathbb{P}\left[ \left|\mathbb{P}\left(\phi(O;\widehat{{\mathbb{P}}}) - \phi(O;\mathbb{P})\mid X\right)\right|\right] \\
      & \lesssim R_{1,n} + \mathbb{P}\left(\big|\widehat{h}^* - h^*\big||\tau|\right) \\
      & \lesssim R_{1,n} + R_{2,n},
  \end{align*}
  as $\widehat{\Delta}_{\delta}^* \leq 1$. We also showed in the proof of Theorem~\ref{thm:est-value}, by Lemma~\ref{lemma:diff} and Assumption~\ref{ass:margin}, that
  \begin{align*}
\left|\mathbb{P}\left((\widehat{\Delta}_{\delta}^* - \Delta_{\delta}^*)\{\phi(O; \mathbb{P}) - q_{1 - \delta}\}\right)\right|
      \lesssim \left(\lVert \widehat{\beta} - \beta\rVert_{\infty} + |\widehat{q}_{1 - \delta} - q_{1 - \delta}|\right)^{1 + b},
  \end{align*}
  which implies that
  \[\sup_{\delta \in [0,1]} \left|\mathbb{P}\left((\phi(O; \mathbb{P}) - q_{1 - \delta})\{\widehat{\Delta}_{\delta}^* - \Delta_{\delta}^*\}\right)\right| \lesssim R_{3,n}^*,\]
  so that $\mathcal{B}_n = O_{\mathbb{P}}(R_{1,n} + R_{2,n} + R_{3,n}^*)$.

\end{proof}


\end{appendices}

\end{document}